\documentclass{article}
\usepackage[english]{babel}
\usepackage{amsmath,amssymb,enumerate,bbm,xcolor,latexsym,theorem}
\usepackage[left=3cm,right=3cm,top=3cm,bottom=3cm]{geometry}
\usepackage{mathtools}
\mathtoolsset{showonlyrefs}

\newcommand\dela[1]{}

\newcommand{\assign}{:=}
\newcommand{\backassign}{=:}
\newcommand{\comma}{{,}}
\newcommand{\infixand}{\text{ and }}
\newcommand{\longhookrightarrow}{{\lhook\joinrel\relbar\joinrel\rightarrow}}
\newcommand{\mathd}{\mathrm{d}}
\newcommand{\mathlambda}{\lambda}
\newcommand{\nobracket}{}

\newcommand{\tmaffiliation}[1]{\\ #1}

\newcommand{\tmemail}[1]{\\ \textit{Email:} \texttt{#1}}

\newcommand{\tmop}[1]{\ensuremath{\operatorname{#1}}}
\newcommand{\tmrsub}[1]{\ensuremath{_{\textrm{#1}}}}
\newcommand{\tmrsup}[1]{\textsuperscript{#1}}
\newcommand{\tmsep}{, }
\newcommand{\tmstrong}[1]{\textbf{#1}}
\newcommand{\tmtextit}[1]{\text{{\itshape{#1}}}}

\newenvironment{enumeratenumeric}{\begin{enumerate}[1.] }{\end{enumerate}}
\newenvironment{itemizedot}{\begin{itemize} }{\end{itemize}}
\newenvironment{proof}{\noindent\textbf{Proof\ }}{\hspace*{\fill}$\Box$\medskip}
\newcounter{nnacknowledgments}

{\theorembodyfont{\rmfamily}\newtheorem{acknowledgments*}[nnacknowledgments]{Acknowledgments}}
\newtheorem{lemma}{Lemma}
\newtheorem{proposition}{Proposition}
{\theorembodyfont{\rmfamily}\newtheorem{remark}{Remark}}
\newtheorem{theorem}{Theorem}
\newcommand{\tmkeywords}{\textbf{Keywords:} }
\newcommand{\tmmsc}{\textbf{A.M.S. subject classification:} }

\begin{document}

\title{
  Decay of correlations in stochastic quantization:\\
  the exponential Euclidean field in two dimensions
}

\author{
  Massimiliano Gubinelli
  \tmaffiliation{Mathematical Institute \\
  University of Oxford \\
  United Kingdom}
  \tmemail{gubinelli@maths.ox.ac.uk
  }
  \and
  Martina Hofmanov{\'a}
  \tmaffiliation{Fakult{\"a}t f{\"u}r Mathematik\\
  Universit{\"a}t Bielefeld \\
  Germany}
  \tmemail{hofmanova@math.uni-bielefeld.de}
  \and
  Nimit Rana
  \tmaffiliation{Fakult{\"a}t f{\"u}r Mathematik\\
  Universit{\"a}t Bielefeld \\
  Germany\\
  \tmtextit{and}\\
  Department of Mathematics\\
  Imperial College London\\
  United Kingdom}
  \tmemail{n.rana22@imperial.ac.uk}
}

\date{}

\maketitle

\begin{abstract}
  We present two approaches to establish the exponential decay of correlation
  functions of Euclidean quantum field theories (EQFTs) via stochastic
  quantization (SQ). In particular we consider the elliptic stochastic
  quantization of the H{\o}egh--Krohn (or $\exp (\alpha \phi)_2$) EQFT in two
  dimensions. The first method is based on a path-wise coupling argument and
  PDE apriori estimates, while the second on estimates of the Malliavin
  derivative of the solution to the SQ equation.
\end{abstract}

\tmmsc{60H17}{\tmsep}{60H07}{\tmsep}{81T07}

\tmkeywords{stochastic quantization}{\tmsep}{H{\o}egh-Krohn
model}{\tmsep}{decay of correlations}{\tmsep}{Euclidean quantum field theory}

{\tableofcontents}

\section{Introduction}
\qquad The last decade has seen a renewed interest in the study of rigorous
stochastic quantization (SQ) of Euclidean quantum field theories (EQFTs). SQ
is a technique, first proposed by
Nelson~{\cite{nelsonDynamicalTheoriesBrownian1967}} and
Parisi--Wu~{\cite{PW81}}, to realize EQFTs, or more generally Gibbsian
measures on $\mathbb{R}^d$ obtained as limits of perturbations of Gaussian
measures, as solutions to certain stochastic partial differential equations (SPDEs) driven by Gaussian noise. After the pioneering work of Jona--Lasinio and
Mitter~{\cite{JLM85,JLM90}} and
Da~Prato--Debussche~{\cite{dapratoStrongSolutionsStochastic2003}}, only very
recently substantial advances have allowed to attack the challenging problem
of the SQ for classical EQFTs, including the $\Phi^4_3$ model, see
e.g.~{\cite{H14,K16,mourratGlobalWellposednessDynamic2017,CC18,albeverioInvariantMeasureFlow2018,GH2019,moinatSpaceTimeLocalisationDynamic2020,gubinelliPDEConstructionEuclidean2021,gubinelliStochasticQuantisationFractional2023}}.

While the original approach of Parisi--Wu to the SQ method based on a Langevin equilibrium diffusion gives rise to parabolic SPDEs, this it is not the only possibility. Nowadays we dispose of at least two other methods of stochastic quantization:
\begin{itemize}
  \item \tmtextit{the elliptic SQ
  approach}~{\cite{AVG20,AVG2021,BV2021arXiv,GH2019}}, based on the
  dimensional reduction phenomenon described by Parisi and
  Sourlas~{\cite{PS79,PS82}} and involving the solutions of an elliptic
  singular SPDE in $d + 2$ dimensions;
  
  \item \tmtextit{the variational
  method}~{\cite{barashkovVariationalMethodPhi2020,barashkovStochasticControlApproach2022,barashkovVariationalMethodEuclidean2021}}
  which involves forward--backward SDEs and can be also applied to fermionic
  EQFTs~{\cite{devecchiStochasticAnalysisSubcritical2022}}.
\end{itemize}

The aim of this work is to discuss the decay of correlations of Euclidean
quantum fields from the point of view of the SQ methods. In particular we
consider the elliptic SQ framework and restrict our attention to the following
elliptic SQ equation with respect to the real valued random field $\varphi
(z), z \in \mathbb{R}^4$,
\begin{equation}
  (- \Delta_{} + m^2) \varphi + \alpha \exp (\alpha \varphi - \infty) = \xi,
  \label{intro-renorm-SPDE}
\end{equation}
where $\alpha \in \mathbb{R}$ and $m > 0$. Here, $\xi$ is a Gaussian white
noise on $\mathbb{R}^4$ and $- \infty$ means that the equation should be
properly renormalized. The existence of a unique solution to
equation~\eqref{intro-renorm-SPDE} and the link with the corresponding EQF
measure in two dimensions, called the H{\o}egh--Krohn model~{\cite{HK71}} (also
known as Liouville model in the literature) has been established
in~{\cite{AVG2021}} for
\[ | \alpha | < \alpha_{\max} \assign 4 \pi \sqrt{8 - 4 \sqrt{3}} . \]
More precisely, well-posedness holds in the weighted Besov space $B_{p, p,
\ell}^s (\mathbb{R}^4)$, for suitable  ($p, s$)  given
in~\eqref{rangeOfParams} and $\ell > 0$ large enough (see
Section~\ref{sec-notation} for precise notations). 

The estimation of connected (or truncated) correlation functions, for
example, the connected two-point function,
\begin{equation}
  \begin{array}{l}
    \mathbb{E} [\varphi (x_1) \varphi (x_2)] -\mathbb{E} [\varphi (x_1)]
    \mathbb{E} [\varphi (x_2)], \quad x_1, x_2 \in \mathbb{R}^4,
  \end{array} \label{intro-corr}
\end{equation}
is a basic goal of any constructive EQFT approach. General truncated
correlation functions allow to infer informations about masses of the
particles in the QFT and estimate scattering amplitudes (see
e.g.~{\cite{H65}}).
In the constructive literature, estimation of the connected correlation
functions is obtained via cluster expansion methods or correlation
inequalities. See for example the early work of
Glimm--Jaffe--Spencer~{\cite{GJS73,GJS74}}. The literature about expansion
methods abounds. We suggest the interested reader to refer
to~{\cite{GlimmJaffe87,APS09,AKHZ89,DCGR20}} and the reference therein for
details and to~{\cite{J00}} for a nice review of related results. Expansion
methods for Euclidean fields involve two primary steps. The initial step is to expand the interaction into parts localized in different bounded volumes of Euclidean space. This gives control over the infinite volume method to
establish the exponential decay of correlations. The second step is to expand interaction into components which are localized on different momentum scales. This helps in dealing with the local regularity properties of correlation functions. The technical difficulty is to mix these two expansions in a manageable way and to systematically extract contributions which require renormalization.  Correlation inequalities methods instead employ discrete approximations, such as lattice approximations, whose specific algebraic properties allow for establishing bounds on a sufficiently broad class of observables.

While expansion methods can be applied to stochastic quantization, as evidenced in works such as {\cite{D90,JLS96}}, we look here for a \tmtextit{stochastic analytic approach} leveraging the intrinsic features of SQ. Parisi {\cite{P81}} presented an early non-rigorous discussion of correlations within the SQ approach and studied how to estimate them directly via computer simulations. In this paper we introduce two simple, general and direct methods to study correlations in SQ applying them to the elliptic SQ of the exponential model~\eqref{intro-renorm-SPDE}:
\begin{description}
  \item[Coupling approach] It is possible to infer the decay of truncated
  correlations by proving that the solutions to the SQ equation exhibit almost  independent behaviour in different regions of space. This can be achieved by coupling the solution to two independent copies by suitably choosing the driving noises. As far as our knowledge extends, it has been Funaki {\cite{funakiReversibleMeasuresMultidimensional1991}} who first
  introduced this idea in the context of equilibrium dynamics of
  Ginzburg--Landau continuum models.
  
  \item[Malliavin calculus approach] Parisi~{\cite{P81}} suggests to study
  variations of the SQ equations in order to infer truncated two-point
  correlations. His observation can actually be made precise and more general using the stochastic calculus of variations, i.e. the Malliavin
  calculus {\cite{Nualart2006}}, and computing derivatives of the solutions to  the SQ equation w.r.t. the driving noise $\xi$.
\end{description}
These two approaches will be used to prove the following statement about a
general class of truncated covariances:

\begin{theorem}
  \label{thm-main}Let $F_1, F_2$ be Lipschitz and functionals on $B_{p, p, \ell}^s (\mathbb{R}^4)$ and $f$ be a given smooth
  function supported in an open ball of unit radius around the origin. Then we  have the following exponential decay
  \begin{equation}
    | \tmop{Cov} (F_1 (f \cdot \varphi (\cdot + x_1)), F_2 (f \cdot \varphi
    (\cdot + x_2))) | \leqslant M  e^{- c | x_1 - x_2 |}, \label{corrDecay}
  \end{equation}
  for all $x_1, x_2 \in \mathbb{R}^4$ where the constant $M$ depends on $m, f, F_1, F_2$, the constant $c$ depends on $m$ but both are independent of $x_1, x_2$.
\end{theorem}

\begin{remark}
  Here $\tmop{Cov} (F, G) \assign \mathbb{E} [F G] -\mathbb{E} [F] \mathbb{E}[G]$ as usual and $(f \cdot \varphi (\cdot + x)) (\phi)  := \varphi (f (\cdot - x) \phi (\cdot))$ for every test function $\phi$.
\end{remark}
In particular we prove that the solution of SQE~\eqref{intro-renorm-SPDE}
satisfies (formally),
\begin{equation}
  | \tmop{Cov} (\varphi (x_1), \varphi (x_2)) | \lesssim e^{- c | x_1 - x_2
  |}, \forall x_1, x_2 \in \mathbb{R}^4 . \label{intro-corr-est}
\end{equation}
It follows from Theorem~\ref{thm-main} that the exponential EQFT in two dimensions has a mass gap, a fact first proven in {\cite{AHK74}} via correlation inequalities for the lattice approximation.

These approaches are general enough to be applicable to other EQFT models
like $P(\varphi)_2$ or $\Phi^4_3$ models. However a fundamental difficulty
presents itself in establishing the required apriori estimates for the
coupling method or in controlling the decay of Malliavin derivative in the
Malliavin method. Both these difficulties originate in the lack of convexity
of the renormalized interaction for a general EQFT. A similar problem is
present in the analysis of logarithmic Sobolev inequalities for EQFT in
bounded volumes {\cite{bauerschmidtLogSobolevInequality2021,bauerschmidtLogSobolevInequalityVarphi2022}} especially for polynomial models. It also manifests in controlling the infinite volume limit of EQFT via stochastic
quantization {\cite{GH2019,gubinelliPDEConstructionEuclidean2021,barashkovVariationalMethodEuclidean2021,gubinelliStochasticQuantisationFractional2023}}, leading to a major obstacle in establishing uniqueness of the infinite volume solutions to the SQ equation.

Fortunately, these difficulties do not show up in the exponential model
because its renormalization is multiplicative and it does not spoil the convex
character of the interaction. For this reason our methods could be readily
applied to obtain decay of correlations for the Sinh--Gordon model studied
in {\cite{BV2021arXiv}}. Another model where uniqueness and correlations can
be controlled via stochastic quantization is the Sine--Gordon model (for large
mass and up the first renormalization threshold), studied by Barashkov via the
variational method in {\cite{barashkovStochasticControlApproach2022}}. Let us also mention that, inspired by the present paper, the coupling method has been already used to show decay of correlations for Euclidean fermionic QFTs and and for sine-Gordon Euclidean QFTs via the FBSDE SQ method, respectively in {\cite{devecchiStochasticAnalysisSubcritical2022}} and {\cite{GubinelliMeyer2024}}.

Let us stress that proving uniqueness of (any kind of) stochastic
quantization and establishing decay of correlation of models like $\Phi^4_{2,
3}$ at high temperature is still largely an open problem which should be
considered, in our opinion, as a crucial test to evaluate the effectivity of
stochastic quantization as a constructive tool in quantum field theory. The
present work is a preliminary step in the direction of understanding better
this problem, and in general in devising appropriate tools to study
stochastically quantized EQFTs.

\begin{acknowledgments*}
  Most of this work was written when the last author was a post-doctoral fellow at Universit{\"a}t Bielefeld, Germany. The financial support from the German Science Foundation DFG through the Research Unit FOR 2402 is gratefully acknowledged.
\end{acknowledgments*}

{\noindent}{\tmstrong{Plan of the paper. }}After introducing notations and definitions of function spaces in Section~\ref{sec-notation}, the paper is structured into two main parts. In Section~\ref{sec-coupling}, we present a proof of Theorem~\ref{thm-main} utilizing the coupling method, commencing with a review of essential results from~{\cite{AVG2021}}. Following this, in Section~\ref{sec-malliavin}, we provide the Malliavin calculus proof of Theorem~\ref{thm-main}, beginning with a summary of relevant tools. The paper concludes with Appendix~\ref{sec-appendix}, where we revisit a few necessary results from the literature and establish the existence and uniqueness of solutions to the approximate equation~\eqref{SPDE-Kregapprox0}.{\medskip}

\subsection{Notations}\label{sec-notation}
In this section we describe some notations and definitions of function spaces used across the whole paper. Some approach depending notations which are also used in the paper are discussed in the corresponding sections.
\begin{itemizedot}
  \item Throughout the paper, we use the notation $a \lesssim b$ if there exists a constant $c > 0$, independent of the variables under consideration, such that $a \leqslant c b$. If we want to emphasize the dependence of $c$ on the variable $x$, then we write $a (x) \lesssim_x b (x)$. The symbol $\assign$ means that the right hand side of the equality defines the left hand side.
  
  \item We set $\mathcal{L} \assign - \Delta + m^2$.
  
  \item For a distribution $\varphi$, a smooth function $f$ and $x \in
  \mathbb{R}^d$, we define the translated distribution $(\varphi (\cdot + x))(\phi) = \varphi (\phi (\cdot - x))$ for all test functions $\phi$ and by $f\cdot \varphi (\cdot + x)$ we denote the multiplication of a smooth function $f$ and distribution $\varphi (\cdot + x)$.
  
  \item By $\mathbb{N}$ we understand the set of natural numbers $\{ 1, 2, \ldots \}.$ For $k \in \mathbb{N} \cup \{ 0 \}$, we write $C^k  (\mathbb{R}^d)$ to denote the set of real valued functions which are
  differentiable up to $k$-times and the $k$-th derivative is continuous. We write $C (\mathbb{R}^d)$ for $k = 0$ and the topology we consider on this space is uniform norm topology. By $C_c^k (\mathbb{R}^4)$ we mean the collection of functions in $C^k  (\mathbb{R}^d)$ having compact support. We denote the the space of smooth functions having compact support by $C_c^{\infty} (\mathbb{R}^4)$. 
  
  \item For any $\ell > 0$ and weight $r_{\ell} (x) \assign (1 + | x |^2)^{-
  \ell / 2}$, by $C_{\ell}^0 (\mathbb{R}^d)$ we denote the space of continuous functions on $\mathbb{R}^d$ such that
  \[ \| f \|_{C_{\ell}^0} {\assign \sup_{x \in \mathbb{R}^d}} 
     | f (x) r_{\ell} (x) | < \infty . \]
  \item By symbol $L^p_{\ell} (\mathbb{R}^d), p \in [1, \infty],$ we mean the
  Banach space of all (equivalence classes of) $\mathbb{R}$-valued weighted
  $p$-integrable functions on $\mathbb{R}^d$. The norm in $L^p_{\ell}
  (\mathbb{R}^d), 1 \leqslant p < \infty$ is given by
  \[ \| f \|_{L^p_{\ell}} \assign \left[ \int_{\mathbb{R}^d} |
     f (y) r_{\ell} (y) |^p \mathd y \right]^{1 / p}, \qquad f \in L^p_{\ell}
     (\mathbb{R}^d) . \]
  For $p = \infty$ we understand it with the usual modification. If $\ell = 0$
  we only write $L^p (\mathbb{R}^d)$ instead $L^p_0 (\mathbb{R}^d)$. Sometimes
  we also use weight function $r_{\lambda, \ell} (x) \assign (1 + \lambda | x
  |^2)^{- \ell / 2}$, for $\lambda, \ell >0$,  and in this case we define $L^p_{\lambda, \ell}
  (\mathbb{R}^d)$ by writing $r_{\lambda, \ell}$ in place of $r_{\ell}$ in
  definition of $L^p_{\ell} (\mathbb{R}^d)$. Similarly we define $L^p_{\ell}(E)$ and $L^p_{\lambda, \ell}(E)$ for an open subset $E \subset \mathbb{R}^d$. 
  
  \item Let $s$ be a real number and $(p, q)$ be in $[1, \infty]^2$. The
  weighted Besov space $B_{p, q, \ell}^s (\mathbb{R}^d)$ consists of all
  tempered distributions $f \in \mathcal{S}' (\mathbb{R}^d)$ such that the
  norm
  \[ \| f \|_{B_{p, q, \ell}^s} \assign \left[ \sum_{j \geqslant - 1} 2^{s j
     q} \| \Delta_j (f) \|_{L^p_{\ell} (\mathbb{R}^d)}^q \right]^{1 / q} \]
  is finite, where $\Delta_j$ are the non-homogeneous dyadic blocks. See
  Appendix A of~{\cite{AVG2021}} for details and properties of $B_{p, q,
  \ell}^s (\mathbb{R}^d)$. We set $C_{\ell}^2(\mathbb{R}^d) := B_{\infty, \infty,
  \ell}^2 (\mathbb{R}^d)$. 
  
  \item For $r > 0, x \in \mathbb{R}^d,$ we denote an open ball of radius $r$
  around $x$ by $B (x, r)$. We also use $d (x, S)$ to define the distance
  between the point $x \in \mathbb{R}^d $ and set $S \subset \mathbb{R}^d .$
  
  \item Let $\mathfrak{a}$ be an auxiliary (radial) smooth, compactly
  supported function such that $\tmop{supp} \mathfrak{a} \subset B (0, 1)$,
  $\int \mathfrak{a} (x) d x = 1,$ and $\mathfrak{a}_{\varepsilon} (x) \assign
  \varepsilon^{- 4} \mathfrak{a} (x / \varepsilon), x \in \mathbb{R}^4$. Note
  that $\tmop{supp} \mathfrak{a}_{\varepsilon} \subset B (0, \varepsilon) .$
\end{itemizedot}
Note that to save space we do not write the integration limit and the measure in the case when it is easily understood from the context.

\section{The coupling approach}\label{sec-coupling}
In this approach towards to proof of Theorem~\ref{thm-main} we first
prove~\eqref{corrDecay} for a random field $\varphi_{\varepsilon}$ which solves an approximation~\eqref{approxSPDE} of SPDE~\eqref{intro-renorm-SPDE}. Then due to Fatou's lemma we pass to the limit $\varepsilon \rightarrow 0$ and obtain \eqref{corrDecay} for $\varphi$. We only need to consider the case of large $l \assign | x_1 - x_2 |$ in detail as for small $l$ the estimate \eqref{corrDecay} holds trivially.

Let us now sketch briefly the idea of the coupling approach. We consider two open balls $D_1$ and $D_2$ in $\mathbb{R}^4$ of radius $l / 2$ with centers
$x_1$ and $x_2$. Further, we take
two copies of Gaussian independent space white noises $\zeta_1$ and $\zeta_2$
and define, for $i = 1, 2$, \
\[ \xi_i \assign \mathbbm{1}_{D_i} \xi +\mathbbm{1}_{D_i^c} \zeta_i . \]
In this way, in $D_i$ we have that $\xi = \xi_i$ for $i = 1, 2$, while $\xi_1$
and $\xi_2$ are independent everywhere. We let $X_{\varepsilon}, X_{1,
\varepsilon}$ and $X_{2, \varepsilon}$ be the solutions to linear part (cfr.
\eqref{OU}) of the approximations of the equation~\eqref{intro-renorm-SPDE}
with noises replaced by $\xi_{\varepsilon}$, $\xi_{1, \varepsilon}$ and
$\xi_{2, \varepsilon}$, respectively. Therefore $X_{1, \varepsilon}$ and
$X_{2, \varepsilon}$ are independent while we will have $X_{i, \varepsilon}
\approx X_{\varepsilon}$ in $D_i$. By stability estimates for
eq.~\eqref{intro-renorm-SPDE} we can derive estimates of the form
(cfr.~\eqref{eqn12}) \
\begin{align*}
	 \mathbb{E} & [\| f \cdot \varphi_{\varepsilon} (\cdot + x_i) - f \cdot
	\varphi_{i, \varepsilon} (\cdot + x_i) \|_{B_{p, p,
			\ell}^s}^{\mathfrak{p}}]  \\
		& \qquad \lesssim  e^{-\mathfrak{c} \left( 1 -
			\frac{l}{8} \right)} (\mathbb{E} [\| X_{\varepsilon} - X_{i, \varepsilon}
		\|_{L^{\mathfrak{p}}}^{\mathfrak{p}}] +\mathbb{E} [\|
		\bar{\varphi}_{\varepsilon} \|_{L^{\mathfrak{p}}}^{\mathfrak{p}}]
		+\mathbb{E} [\| \bar{\varphi}_{i, \varepsilon}
		\|_{L^{\mathfrak{p}}}^{\mathfrak{p}}]
\end{align*}
for some $\mathfrak{c}$ which depends on $m$ and $\mathfrak{p}$, where
$\mathfrak{p} \in [2, \infty)$ is fixed. In the above we have
$\varphi_{\varepsilon} = \bar{\varphi}_{\varepsilon} + X_{\varepsilon}$ and
$\varphi_{i, \varepsilon} = \bar{\varphi}_{i, \varepsilon} + X_{i,
\varepsilon}$, where ${\varphi}_{\varepsilon}$ and ${\varphi}_{i,
\varepsilon}$ respectively, are the unique solutions to the regularized SPDE
\eqref{approxSPDE} with $\xi_{\varepsilon}$ and $\xi_{i, \varepsilon}$ as detailed in Subsection \ref{subsec-coupling-notation}. This
estimate allows to replace $\varphi_{\varepsilon}$ by $\varphi_{i,
\varepsilon}$ in $D_i$ by paying a small error of the order $e^{- c l}$ for some $c>0$ (independent of $x_i$). Since
$\varphi_{1, \varepsilon}$ and $\varphi_{2, \varepsilon}$ are independent,
from the last estimate we can conclude easily the exponential decay for
Lipschitz observables, see Subsection \ref{subsec-coupling-proof} for details.

\subsection{Preliminaries}\label{subsec-coupling-notation}

In this subsection we summarize the steps, with another suitably modified
approximation, of the proof from~{\cite{AVG2021}}, which also set further
required notation. The main result of~{\cite{AVG2021}}, which is about the
existence of a unique solution to the singular SPDE~\eqref{intro-renorm-SPDE},
is based on the Da Prato--Debussche
trick~{\cite{dapratoStrongSolutionsStochastic2003}} and the fact that the Wick
exponential is a positive measure.
\begin{itemizedot}
  \item Let us consider a complete probability space $(\Omega, \mathfrak{F},
  \mathbb{P})$, which satisfies the usual hypothesis, and $\xi$ as Gaussian
  white noise on $\mathbb{R}^4$ defined on $(\Omega, \mathfrak{F},
  \mathbb{P})$.
  
  \item Let $X$ be the solution to $\mathcal{L}X = \xi$. The existence and
  uniqueness of such $X \in B_{q, q, \ell}^{- \delta} (\mathbb{R}^4)$ for
  every $q \in [1, \infty], \delta > 0$ and $\ell > 0$ is proved
  in~{\cite{GH2019}}.
  
  \item To avoid clumsy notation we write $\eta \assign \exp^{\diamond}
  (\alpha \mathcal{L}^{- 1} \xi)$ for the renormalized version of the
  distribution $\exp (\alpha \mathcal{L}^{- 1} \xi - \infty)$, where
  $\exp^{\diamond}$ denotes the Wick exponential of the Gaussian distribution
  $X =\mathcal{L}^{- 1} \xi$.
  
  \item The first step in giving a meaning to equation
  \eqref{intro-renorm-SPDE} is to take the decomposition $\varphi =
  \bar{\varphi} + X$. Then observe that formally $\bar{\varphi} $ satisfies
  \begin{equation}
    \begin{array}{l}
      \mathcal{L} \bar{\varphi} + \alpha \exp (\alpha \bar{\varphi}) \eta = 0.
    \end{array} \label{goodSPDE}
  \end{equation}
  \item For any $\varepsilon > 0$ let us set $\xi_{\varepsilon} \assign
  \mathfrak{a}_{\varepsilon} \ast \xi$ where $\ast$ denotes convolution. Note
  that
  \[ \begin{array}{l}
       \eta = \sum_{k = 0}^{\infty} \frac{\alpha^k}{k!}  (\mathcal{L}^{- 1}
       \xi)^{\diamond k},
     \end{array} \]
  where $\diamond$ denotes the Wick product and $(\mathcal{L}^{- 1} \xi)^{\diamond
  k}$= $\underbrace{\mathcal{L}^{- 1} \xi \diamond \mathcal{L}^{- 1} \xi
  \diamond \cdots \diamond \mathcal{L}^{- 1} \xi}_{k - \tmop{times}} =
  X^{\diamond k}$. By denoting $X_{\varepsilon} =\mathcal{L}^{- 1}
  \xi_{\varepsilon}$ as the unique smooth solution to
  $\mathcal{L}X_{\varepsilon} = \xi_{\varepsilon}$, we set
  $\eta_{\varepsilon}$ as the following positive measure
  \begin{equation}\label{def-eta-eps}
      \eta_{\varepsilon} (\mathd z) = \exp^{\diamond} (\alpha \mathcal{L}^{-
       1} \xi_{\varepsilon}) \mathd z = \exp (\alpha \mathcal{L}^{- 1}
       \xi_{\varepsilon} - C_{\varepsilon}) \mathd z,
  \end{equation}
  where $C_{\varepsilon} \assign \frac{\alpha^2}{2} \mathbb{E} [|
  X_{\varepsilon} |^2]$.
  
  Moreover, from Section 3.1 of~{\cite{AVG2021}}, we know that
  \[ \begin{array}{l}
       \eta_{\varepsilon} = \sum_{k = 0}^{\infty} \frac{\alpha^k}{k!} 
       (\mathcal{L}^{- 1} \xi_{\varepsilon})^{\diamond k},
     \end{array} \]
  and, for $| \alpha | < 4 \sqrt{2} \pi, p \in (1, 2], s \leqslant -
  \frac{\alpha^2 (p - 1)}{(4 \pi)^2}$ and $\ell > 0$ large enough, \
  $\eta_{\varepsilon} \rightarrow \eta,$ as $\varepsilon \rightarrow 0$, in
  probability in $B_{p, p, \ell}^s (\mathbb{R}^4) .$ Note that the convergence
  $\eta_{\varepsilon} \rightarrow \eta$ in probability implies that there
  exists a sequence $\varepsilon_n$, which converges to $0$, such that
  $\eta_{\varepsilon_n} \rightarrow \eta,$ as $\varepsilon_n \rightarrow 0,$
  in $B_{p, p, \ell}^s (\mathbb{R}^4)$ $\mathbb{P}$-almost surely. We will fix
  this sequence $\{ \varepsilon_n \}_{n \geqslant 1}$ in the whole paper. \
  
  \item By Theorems~21 and~25 from~{\cite{AVG2021}} we have that for any $|
  \alpha | < \alpha_{\max}$, there exist $p, s, \delta$ satisfying
  \begin{equation}
    1 < p \leqslant 2, \qquad p < \frac{2 (4 \pi)^2}{\alpha^2}, \qquad - 1 < s \leqslant -
    \frac{\alpha^2 (p - 1)}{(4 \pi)^2}  \quad \infixand \quad 0 < \delta < s + 1,
    \label{rangeOfParams}
  \end{equation}
  the equation \eqref{goodSPDE} has a unique solution $\bar{\varphi}$ in
  $B_{p, p, \ell + \delta'}^{s + 2 - \delta} (\mathbb{R}^4)$, $\mathbb{P}$-almost surely,
  for large enough $\ell > 0$ and small enough $\delta' > 0$. Moreover,
  \begin{equation}
    \begin{array}{l}
      \alpha \bar{\varphi} \leqslant 0
    \end{array} \label{neg-phi-bar}
  \end{equation}
  holds true. Furthermore, for $\{ \varepsilon_n \}_{n \geqslant 1}$ as fixed
  above, $\bar{\varphi}_{\varepsilon_n} \rightarrow \bar{\varphi} $ in $B_{p,
  p, \ell + \delta'}^{s + 2 - \delta} (\mathbb{R}^4)$ as $n \rightarrow
  \infty$, $\mathbb{P}$-almost surely, where $\bar{\varphi}_{\varepsilon_n}$ solves the
  approximate equation
  \begin{equation}
    \mathcal{L} \bar{\varphi}_{\varepsilon_n} + \alpha \exp (\alpha
    \bar{\varphi}_{\varepsilon_n}) \eta_{\varepsilon_n} = 0
    \label{good-approx-SPDE}
  \end{equation}
  uniquely in $C_{\ell}^0 (\mathbb{R}^4)$ such that $\alpha
  \varphi_{\varepsilon_n} \leqslant 0$.
  
  \item Thus, for $(p, s)$ such that \eqref{rangeOfParams} holds and $\ell >
  0$ large enough, $\varphi = X + \bar{\varphi} \in B_{p, p, \ell}^s
  (\mathbb{R}^4),$ $\mathbb{P}$-almost surely, solves SPDE \eqref{intro-renorm-SPDE}
  uniquely. If we consider the following approximation of
  SPDE~\eqref{intro-renorm-SPDE}
  \begin{equation}
    \begin{array}{l}
      \mathcal{L} \varphi_{\varepsilon_n} + \alpha \exp (\alpha
      \varphi_{\varepsilon_n} - C_{\varepsilon_n})
      =\mathfrak{a}_{\varepsilon_n} \ast \xi,
    \end{array} \label{approxSPDE}
  \end{equation}
  then, from the proof of Theorem 35 of~{\cite{AVG2021}}, we know that
  $\varphi_{\varepsilon_n} = \bar{\varphi}_{\varepsilon_n} +
  X_{\varepsilon_n}$ is the unique solution to~\eqref{approxSPDE} and
  $\varphi_{\varepsilon_n} \rightarrow \varphi$ in $B_{p, p, \ell}^s
  (\mathbb{R}^4)$, $\mathbb{P}$-almost surely as $n \rightarrow \infty$. 
\end{itemizedot}
Let us recall that we have fixed the sequence of $\{ \varepsilon_n \}_{n \in
\mathbb{N}}$ which converges to $0$ as $n \rightarrow \infty$. To shorten the
notation, we will write $\varepsilon \rightarrow 0$ equivalently to $n
\rightarrow \infty$.

\subsection{Proof of Theorem \ref{thm-main}}\label{subsec-coupling-proof}

\qquad Assume that $| x_1 - x_2 | \leqslant 8.$ It is trivial to
get~\eqref{corrDecay} because its l.h.s. is bounded. Consider now the
complementary case and let $l \assign | x_1 - x_2 | > 8$. Take two open balls
$D_1$ and $D_2$ in $\mathbb{R}^4$ of radius $l / 2$ with centers $x_1$ and
$x_2,$ respectively. \ Further, we take two copies of Gaussian independent
space white noises $\zeta_1$ and $\zeta_2$ defined on $(\Omega, \mathfrak{F},
\mathbb{P})$. Define the processes $X_1$ and $X_2$ as follows:
\begin{equation}
  \quad \begin{array}{l}
    \mathcal{L}X_1 =\mathbbm{1}_{D_1} \xi +\mathbbm{1}_{D_1^c} \zeta_1
    \backassign \xi_1, \qquad \textrm{\tmop{and}} \qquad \mathcal{L}X_2
    =\mathbbm{1}_{D_2} \xi +\mathbbm{1}_{D_2^c} \zeta_2 \backassign \xi_2 .
    \label{OU}
  \end{array}
\end{equation}
Note that that $\xi - \xi_i = 0$ on $D_i, i = 1, 2$ in the sense of
distributions $\mathbb{P}$-a.s. Moreover, since \ $D_1 \cap D_2 = \emptyset$,
the processes $X_1$ and $X_2$ are independent. Indeed, by setting
$(\mathbbm{1}_{D_i} \xi) (f) \assign \xi (\mathbbm{1}_{D_i} f)$, we observe
that for $f, g \in L^2 (\mathbb{R}^4),$
\[ \mathbb{E} [\xi_1 (f) \xi_2 (g)] = \langle \mathbbm{1}_{D_1} f,
   \mathbbm{1}_{D_2} g \rangle_{L^2} = 0. \]

Let us set $\xi_{\varepsilon} \assign \mathfrak{a}_{\varepsilon} \ast \xi$ and
$\xi_{i, \varepsilon} \assign \mathfrak{a}_{\varepsilon} \ast \xi_i, i = 1, 2$
for the whole subsection. Let $\varphi_{\varepsilon}$ and $\varphi_{i,
\varepsilon}$, respectively, be the unique solutions to the following
regularized version of eq.~\eqref{intro-renorm-SPDE}
\begin{equation}\label{b-approxSPDE}
    \mathcal{L} \varphi_{\varepsilon} + \alpha \exp (\alpha
    \varphi_{\varepsilon} - C_{\varepsilon}) = \xi_{\varepsilon}, 
\end{equation}
and
\begin{equation}\label{b-approxSPDE-1}
    \mathcal{L} \varphi_{i, \varepsilon} + \alpha \exp (\alpha \varphi_{i,
    \varepsilon} - C_{\varepsilon}) = \xi_{i, \varepsilon},
\end{equation}
where $C_{\varepsilon} \assign \frac{\alpha^2}{2} \mathbb{E} [|
X_{\varepsilon} |^2] = \frac{\alpha^2}{2} \mathbb{E} [| X_{i, \varepsilon}
|^2] $ for $X_{\varepsilon} =\mathcal{L}^{- 1} \xi_{\varepsilon}$ and $X_{i,
\varepsilon} =\mathcal{L}^{- 1} \xi_{i, \varepsilon}, i = 1, 2$. Note that due
to stationarity in space of the white noise $\xi$, the constant
$C_{\varepsilon}$ does not depend on $x \in \mathbb{R}^4$.

Next, let us fix $\mathfrak{p} \in [2, \infty)$ and consider $\varphi$, $\varphi_1$ and $\varphi_2$ as the unique solutions to the SPDE~\eqref{intro-renorm-SPDE} with noises $\xi$, $\xi_1$ and $\xi_2$, respectively. Their existence has been summarized in Subsection \ref{subsec-coupling-notation}.  Then observe that,  since $F_1$ and $F_2$ are Lipschitz and bounded functionals, using the H{\"o}lder inequality we get the following
\begin{align}\label{coupling-cov-est1}
    & | \tmop{Cov} (F_1 (f \cdot \varphi (\cdot + x_1)), F_2 (f \cdot \varphi
    (\cdot + x_2))) | \nonumber\\
    & \quad \leq  | \mathbb{E} [(F_1 (f \cdot \varphi (\cdot + x_1)) - F_1 (f
    \cdot \varphi_1 (\cdot + x_1))) F_2 (f \cdot \varphi (\cdot + x_2))] | \nonumber\\
    & \quad \quad + | \mathbb{E} [F_1 (f \cdot \varphi_1 (\cdot + x_1))  (F_2 (f \cdot
    \varphi (\cdot + x_2)) - F_2 (f \cdot \varphi_2 (\cdot + x_2)))] | \nonumber\\
    & \quad \quad + | \mathbb{E} [F_1 (f \cdot \varphi_1 (\cdot + x_1)) F_2 (f \cdot
    \varphi_2 (\cdot + x_2))] -\mathbb{E} [F_1 (f \cdot \varphi (\cdot +
    x_1))] \mathbb{E} [F_2 (f \cdot \varphi (\cdot + x_2))] | \nonumber\\
    & \quad \lesssim_{F_1, F_2}   [\mathbb{E} [\| f \cdot \varphi (\cdot + x_1) - f
    \cdot \varphi_1 (\cdot + x_1) \|_{B_{p, p, \ell}^s}^{\mathfrak{p}}]]^{1
    /\mathfrak{p}}  \nonumber\\
    & \quad \quad + [\mathbb{E} [\| f \cdot \varphi (\cdot + x_2) - f \cdot
    \varphi_2 (\cdot + x_2) \|_{B_{p, p, \ell}^s}^{\mathfrak{p}}]]^{1
    /\mathfrak{p}}.  
\end{align}
Here we used that, since the processes $\xi_1$ and $\xi_2$ are independent and the processes $\xi,\xi_1$ and $\xi_2$ have same law, 
\[ \mathbb{E} [F_1 (f \cdot \varphi_1 (\cdot + x_1)) F_2 (f \cdot \varphi_2
   (\cdot + x_2))] -\mathbb{E} [F_1 (f \cdot \varphi (\cdot + x_1))]
   \mathbb{E} [F_2 (f \cdot \varphi (\cdot + x_2))] = 0. \]
But thanks to Fatou's lemma (see Theorem 2.72~of~{\cite{BahouriEtAl2011}}), to
get \eqref{corrDecay} from \eqref{coupling-cov-est1} it is enough to prove
that, for $i = 1, 2$,
\begin{equation}
  \begin{array}{lll}
    \mathbb{E} [\| f \cdot \varphi_{\varepsilon} (\cdot + x_i) - f \cdot
    \varphi_{i, \varepsilon} (\cdot + x_i) \|_{B_{p, p,
    \ell}^s}^{\mathfrak{p}}] & \lesssim  & e^{- c l},
  \end{array} \label{formal-decay}
\end{equation}
uniform in $\varepsilon$, for some $c > 0$ which does not depend on $x_1,
x_2$.

Due to symmetry, it is sufficient to estimate $\mathbb{E} [\| f \cdot
\varphi_{\varepsilon} (\cdot + x_1) - f \cdot \varphi_{1, \varepsilon} (\cdot
+ x_1) \|_{B_{p, p, \ell}^s}^{\mathfrak{p}}]$. For that let $\tilde{D}_1
\assign B \left( x_1, \frac{l}{4} \right)$ and take 
\[\rho (x) \assign e^{- \beta m | x - x_1 |}, \quad x
\in \mathbb{R}^4, \]  
a weight function where we set the value of $\beta$ later.
Further, let us take $\theta$ as a non-negative smooth function supported in
$\tilde{D}_1$ such that $\theta = 1$ in $\bar{D}_1 \assign B \left( x_1,
\frac{l}{8} \right)$. To shorten the notation we also set $\bar{\rho} (x) :=
\theta (x) \rho (x)$.

Since $f$ has support in $B (0, 1)$, by the Besov
embedding Theorem \ref{thm-BesovEmbedding} followed by continuous embedding of
$L_{\ell}^{\mathfrak{p}}(\mathbb{R}^4)$ into $B^0_{\mathfrak{p}, \infty, \ell}(\mathbb{R}^4)$ we get,
where $\chi_{\varepsilon} \assign \varphi_{\varepsilon} - \varphi_{1,
\varepsilon}$,
\begin{align}\label{eqn0}
    & \| f \cdot \varphi_{\varepsilon} (\cdot + x_1) - f \cdot \varphi_{1,
    \varepsilon} (\cdot + x_1) \|_{B_{p, p, \ell}^s}^{\mathfrak{p}} \lesssim
     \| f (\cdot - x_1) \chi_{\varepsilon} \|_{B_{\mathfrak{p}, p,
    \ell}^s}^{\mathfrak{p}} \nonumber \\
    & \qquad \lesssim \int_{B (x_1, 1)} | f (x - x_1) \chi_{\varepsilon} (x)
    |^{\mathfrak{p}} \mathd x ~\leq~ e^{m\mathfrak{p} \beta} \| f \|_{L^{\infty}} \|  \bar{\rho}
    \chi_{\varepsilon} \|^{\mathfrak{p}}_{L^{\mathfrak{p}} (B (x_1, 1))}. 
\end{align}
\dela{
\begin{equation}
  \begin{array}{lll}
    \| f \cdot \varphi_{\varepsilon} (\cdot + x_1) - f \cdot \varphi_{1,
    \varepsilon} (\cdot + x_1) \|_{B_{p, p, \ell}^s}^{\mathfrak{p}} & \lesssim
    & \| f (\cdot - x_1) \chi_{\varepsilon} \|_{B_{\mathfrak{p}, p,
    \ell}^s}^{\mathfrak{p}}\\
    & \lesssim & \int_{B (x_1, 1)} | f (x - x_1) \chi_{\varepsilon} (x)
    |^{\mathfrak{p}} \mathd x\\
    & \leqslant & e^{m\mathfrak{p} \beta} \| f \|_{L^{\infty}} \|  \bar{\rho}
    \chi_{\varepsilon} \|^{\mathfrak{p}}_{L^{\mathfrak{p}} (B (x_1, 1))} .
  \end{array} 
\end{equation}}
Towards estimating $\| \bar{\rho} \chi_{\varepsilon}
\|^{\mathfrak{p}}_{L^{\mathfrak{p}} (B (x_1, 1))} $, first we claim that
\begin{equation}
  \begin{array}{l}
    \theta (x) (\xi_{\varepsilon} - \xi_{1, \varepsilon}) (x) = 0 \quad
    \tmop{for} \tmop{all} \quad x \in \mathbb{R}^4 .
  \end{array} \label{theta-noise}
\end{equation}
This is obvious for $x \in \mathbb{R}^4 \setminus \tilde{D}_1$. So let us take
$x \in \tilde{D}_1$. Since 
\begin{equation*}
    (\theta (\xi_{\varepsilon} - \xi_{1, \varepsilon})) (x)  =  \theta (x)
   (\mathfrak{a}_{\varepsilon} \ast (\xi - \xi_1)) (x),
\end{equation*}
\dela{\[ \begin{array}{ccc}
     (\theta (\xi_{\varepsilon} - \xi_{1, \varepsilon})) (x) & = & \theta (x)
   \end{array} (\mathfrak{a}_{\varepsilon} \ast (\xi - \xi_1)) (x), \]}
it is sufficient to show that $(\xi - \xi_1, \mathfrak{a}_{\varepsilon} \ast
g)_{\mathcal{S}', \mathcal{S}} = 0$ for all $g \in C_c^{\infty}
(\tilde{D}_1)$, where $(\cdot, \cdot)_{\mathcal{S}', \mathcal{S}}$ is duality
between Schwartz function $\mathcal{S}$ and Schwartz distribution
$\mathcal{S}'$. But, since $\xi - \xi_1 = 0$ on $D_1$, for this it is enough
to show that $\tmop{supp} (\mathfrak{a}_{\varepsilon} \ast g) \subset D_1$.
This follows because
\[ \begin{array}{l}
     (\mathfrak{a}_{\varepsilon} \ast g) (z) = \int_{\tilde{D}_1}
     \mathfrak{a}_{\varepsilon} (z - y) g (y) \, \mathd y,
   \end{array} \]
and for $z \in D_1^c$ and $y \in \tilde{D}_1$, $| z - y | \geqslant
\frac{l}{4} > \frac{\varepsilon l}{4}$. Hence the claim \eqref{theta-noise}.

Next, observe that $\chi_{\varepsilon}$ satisfies
\begin{equation}
  \begin{array}{ll}
    & \mathcal{L} \chi_{\varepsilon} + Q_{\varepsilon} \chi_{\varepsilon} =
    \xi_{\varepsilon} - \xi_{1, \varepsilon},
  \end{array} \label{approxEqnForDiff}
\end{equation}
where $Q_{\varepsilon} \assign \alpha^2 \int_0^1 \exp \{ \alpha \varphi_{1,
\varepsilon} - C_{\varepsilon} + \Theta \alpha (\varphi_{\varepsilon} -
\varphi_{1, \varepsilon}) \} \mathd \Theta > 0$. Then, testing
\eqref{approxEqnForDiff} with $\bar{\rho}^{\mathfrak{p}} |
\chi_{\varepsilon} |^{\mathfrak{p}- 2} \chi_{\varepsilon}$ and
integrating on $\mathbb{R}^4$ give
\begin{equation}
  \int \bar{\rho}^{\mathfrak{p}} | \chi_{\varepsilon} |^{\mathfrak{p}- 2}
  \chi_{\varepsilon} \mathcal{L} \chi_{\varepsilon} + \int
  \bar{\rho}^{\mathfrak{p}} | \chi_{\varepsilon} |^{\mathfrak{p}}
  Q_{\varepsilon} = 0, \label{eqn1}
\end{equation}
where the noise term vanishes because of~\eqref{theta-noise}. The first term
on the l.h.s. above can be expanded as
\begin{align*}
    \int \bar{\rho}^{\mathfrak{p}} | \chi_{\varepsilon} |^{\mathfrak{p}-
     2} \chi_{\varepsilon} \mathcal{L} \chi_{\varepsilon} & = m^2 \int
     \bar{\rho}^{\mathfrak{p}} | \chi_{\varepsilon} |^{\mathfrak{p}} + \int
     \bar{\rho}^{\mathfrak{p}- 1} | \chi_{\varepsilon} |^{\mathfrak{p}- 2}
     \chi_{\varepsilon} (- \Delta (\bar{\rho} \chi_{\varepsilon})) \\
     & \quad +
     \int \bar{\rho}^{\mathfrak{p}- 1} | \chi_{\varepsilon}
     |^{\mathfrak{p}} \Delta \bar{\rho} + 2 \int \bar{\rho}^{\mathfrak{p}- 1} | \chi_{\varepsilon}
     |^{\mathfrak{p}- 2} \chi_{\varepsilon} \nabla \bar{\rho} \cdot \nabla
     \chi_{\varepsilon} .
\end{align*}
But, since the integration by parts and the definition of divergence give
\begin{align*}
    & 2 \int \bar{\rho}^{\mathfrak{p}- 1} | \chi_{\varepsilon}
     |^{\mathfrak{p}- 2} \chi_{\varepsilon} \nabla \bar{\rho} \cdot \nabla
     \chi_{\varepsilon}  =  \frac{2}{\mathfrak{p}} \int
     \bar{\rho}^{\mathfrak{p}- 1} \nabla \bar{\rho} \cdot \nabla |
     \chi_{\varepsilon} |^{\mathfrak{p}}\\
     & \quad = - \frac{2}{\mathfrak{p}} \int | \chi_{\varepsilon}
     |^{\mathfrak{p}} \tmop{div} (\bar{\rho}^{\mathfrak{p}- 1} \nabla
     \bar{\rho}) = - \frac{2}{\mathfrak{p}} \int | \chi_{\varepsilon}
     |^{\mathfrak{p}}  [(\mathfrak{p}- 1) \bar{\rho}^{\mathfrak{p}- 2} |
     \nabla \bar{\rho} |^2 + \bar{\rho}^{\mathfrak{p}- 1} \Delta \bar{\rho}],
\end{align*}
we have
\begin{align}\label{eqn2}
    \int \bar{\rho}^{\mathfrak{p}} | \chi_{\varepsilon} |^{\mathfrak{p}- 2}
    \chi_{\varepsilon} \mathcal{L} \chi_{\varepsilon} & =  m^2 \int
    \bar{\rho}^{\mathfrak{p}} | \chi_{\varepsilon} |^{\mathfrak{p}} + \int
    \bar{\rho}^{\mathfrak{p}- 1} | \chi_{\varepsilon} |^{\mathfrak{p}- 2}
    \chi_{\varepsilon} (- \Delta (\bar{\rho} \chi_{\varepsilon})) \nonumber \\
    & \quad  + \left( 1 - \frac{2}{\mathfrak{p}} \right) \int
    \bar{\rho}^{\mathfrak{p}- 1} | \chi_{\varepsilon} |^{\mathfrak{p}}
    \Delta \bar{\rho} - \frac{2 (\mathfrak{p}- 1)}{\mathfrak{p}} \int |
    \chi_{\varepsilon} |^{\mathfrak{p}}  \bar{\rho}^{\mathfrak{p}- 2} |
    \nabla \bar{\rho} |^2 . 
\end{align}
Thus, substitution of \eqref{eqn2} into \eqref{eqn1} together with
$Q_{\varepsilon} > 0$ yield
\begin{align}\label{eqn3}
    \int \bar{\rho}^{\mathfrak{p}- 1} | \chi_{\varepsilon}
    |^{\mathfrak{p}- 2} \chi_{\varepsilon} (- \Delta (\bar{\rho}
    \chi_{\varepsilon})) & + \left( 1 - \frac{2}{\mathfrak{p}} \right) \int
    \bar{\rho}^{\mathfrak{p}- 1} | \chi_{\varepsilon} |^{\mathfrak{p}}
    \Delta \bar{\rho} + m^2 \int \bar{\rho}^{\mathfrak{p}} |
    \chi_{\varepsilon} |^{\mathfrak{p}} \nonumber\\
    & \leq \frac{2 (\mathfrak{p}- 1)}{\mathfrak{p}} \int |
    \chi_{\varepsilon} |^{\mathfrak{p}}  \bar{\rho}^{\mathfrak{p}- 2} |
    \nabla \bar{\rho} |^2.
\end{align}
Furthermore, since the integration by parts and the product rule of derivative
give
\begin{align}\label{eqn4}
    \int \bar{\rho}^{\mathfrak{p}- 1} | \chi_{\varepsilon} |^{\mathfrak{p}-
    2} \chi_{\varepsilon} (- \Delta (\bar{\rho} \chi_{\varepsilon})) & =
    \int (\mathfrak{p}- 1) | \chi_{\varepsilon} |^{\mathfrak{p}- 2}
    \chi_{\varepsilon} \bar{\rho}^{\mathfrak{p}- 2} \nabla \bar{\rho} \cdot
    \nabla (\bar{\rho} \chi_{\varepsilon}) \nonumber\\
    &  \quad + \int
    \bar{\rho}^{\mathfrak{p}-
    1} | \chi_{\varepsilon} |^{\mathfrak{p}- 2} \nabla
    \chi_{\varepsilon} \cdot \nabla (\bar{\rho} \chi_{\varepsilon}) \nonumber\\
    &  \quad + \int (\mathfrak{p}- 2)
    \bar{\rho}^{\mathfrak{p}- 1} | \chi_{\varepsilon} |^{\mathfrak{p}- 2}
    \nabla \chi_{\varepsilon} \cdot \nabla (\bar{\rho}
    \chi_{\varepsilon}) \nonumber\\
    & =  (\mathfrak{p}- 1) \int | \chi_{\varepsilon} |^{\mathfrak{p}- 2}
    \bar{\rho}^{\mathfrak{p}- 2} | \nabla (\bar{\rho} \chi_{\varepsilon})
    |^2 \nonumber \\
    & \geqslant  0,
\end{align}
from inequality \eqref{eqn3} we obtain
\begin{equation} \label{eqn5}
    \left( 1 - \frac{2}{\mathfrak{p}} \right) \int \bar{\rho}^{\mathfrak{p}-
    1} | \chi_{\varepsilon} |^{\mathfrak{p}} \Delta \bar{\rho} + m^2 \int
    \bar{\rho}^{\mathfrak{p}} | \chi_{\varepsilon} |^{\mathfrak{p}}
    \leqslant \frac{2 (\mathfrak{p}- 1)}{\mathfrak{p}} \int |
    \chi_{\varepsilon} |^{\mathfrak{p}}  \bar{\rho}^{\mathfrak{p}- 2} |
    \nabla \bar{\rho} |^2 .
\end{equation}
Since $\nabla \rho (x) = - m \beta \frac{x - x_1}{| x - x_1 |} \rho (x)$ for
$x \in \mathbb{R}^4 \setminus \{ x_1 \}$,
\[ \begin{array}{l}
     \Delta \bar{\rho} = \rho \Delta \theta + 2 \nabla \rho \cdot \nabla
     \theta + m^2 \beta^2 \theta \rho, \quad \tmop{and} \quad  | \nabla
     \bar{\rho} | \leqslant | \nabla \theta | \rho + m \beta \rho \theta,
   \end{array} \]
 inequality \eqref{eqn5} yield
\begin{align}
    & \left( 1 - \frac{2}{\mathfrak{p}} \right)  \int \bar{\rho}^{\mathfrak{p}-
    1} | \chi_{\varepsilon} |^{\mathfrak{p}} \rho \Delta \theta   -2m \beta \left( 1 - \frac{2}{\mathfrak{p}} \right) \int \bar{\rho}^{\mathfrak{p}-
    1} | \chi_{\varepsilon} |^{\mathfrak{p}} \left[\rho \frac{x - x_1}{| x - x_1 |}  \cdot \nabla
     \theta  \right] \\
     & \qquad \qquad + m^2 \beta^2 \left( 1 - \frac{2}{\mathfrak{p}} \right) \int \bar{\rho}^{\mathfrak{p}-
    1} | \chi_{\varepsilon} |^{\mathfrak{p}}  \theta \rho + m^2 \int
    \bar{\rho}^{\mathfrak{p}} | \chi_{\varepsilon} |^{\mathfrak{p}} \nonumber\\
    & \leq \frac{4 (\mathfrak{p}- 1)}{\mathfrak{p}} \int |
    \chi_{\varepsilon} |^{\mathfrak{p}}  \bar{\rho}^{\mathfrak{p}- 2} \left[| \nabla \theta |^2 \rho^2 + m^2 \beta^2 \rho^2 \theta^2 \right],  
\end{align}
where to get the r.h.s. terms we also used $(a+b)^2 \leq 2(a^2 + b^2)$, $\forall a,b \in \mathbb{R}$. Consequently, by regrouping the terms together with $\bigg\vert \frac{x - x_1}{| x - x_1 |}  \cdot \nabla\theta \bigg\vert \leq |\nabla \theta|$ we get
\begin{align}\label{eqn6}
    & m^2 \left( 1 + \beta^2 \left( \frac{2 - 3\mathfrak{p}}{\mathfrak{p}}
    \right) \right) \| \bar{\rho}  \chi_{\varepsilon}
    \|_{L^{\mathfrak{p}}}^{\mathfrak{p}} + \left( 1 - \frac{2}{\mathfrak{p}}
    \right) \int \bar{\rho}^{\mathfrak{p}- 1} | \chi_{\varepsilon}
    |^{\mathfrak{p}} \rho \Delta \theta \nonumber\\
    & \quad \leq  ~ \frac{4 (\mathfrak{p}- 1)}{\mathfrak{p}} \int | \nabla \theta
    |^2 \rho^{\mathfrak{p}} \theta^{\mathfrak{p}- 2}  | \chi_{\varepsilon}
    |^{\mathfrak{p}} + 2 m \beta \left( 1 - \frac{2}{\mathfrak{p}} \right)
    \int \bar{\rho}^{\mathfrak{p}- 1} | \chi_{\varepsilon} |^{\mathfrak{p}}
    \rho | \nabla \theta | .
\end{align}
Moreover, since $\theta$ is supported in $\tilde{D}_1$ and $\theta = 1$ on $\bar{D}_1$, \eqref{eqn6}
gives
\begin{align}\label{eqn8}
     m^2 \left( 1 + \beta^2 \left( \frac{2 - 3\mathfrak{p}}{\mathfrak{p}}
    \right) \right) \| \bar{\rho}  \chi_{\varepsilon} \|_{L^{\mathfrak{p}}
    (\tilde{D}_1)}^{\mathfrak{p}} & \leq \frac{4 (\mathfrak{p}-
    1)}{\mathfrak{p}} \int_{\tilde{D}_1 \setminus \bar{D}_1} | \nabla \theta
    |^2 \rho^{\mathfrak{p}} \theta^{\mathfrak{p}- 2}  | \chi_{\varepsilon}
    |^{\mathfrak{p}} \nonumber\\
    &  \quad + 2 m \beta \left( 1 - \frac{2}{\mathfrak{p}} \right)
    \int_{\tilde{D}_1 \setminus \bar{D}_1} \bar{\rho}^{\mathfrak{p}- 1} |
    \chi_{\varepsilon} |^{\mathfrak{p}} \rho | \nabla \theta | \nonumber\\
    &  \quad  + \left( 1 - \frac{2}{\mathfrak{p}} \right) \int_{\tilde{D}_1
    \setminus \bar{D}_1} \bar{\rho}^{\mathfrak{p}- 1} | \chi_{\varepsilon}
    |^{\mathfrak{p}} \rho | \Delta \theta | .
\end{align}
To keep the coefficient of $\| \bar{\rho}  \chi_{\varepsilon} \|_{L^{\mathfrak{p}}}$ positive in the l.h.s. above, we choose $\beta = \beta (\mathfrak{p}) > 0$ so small such that
\begin{equation}
  \begin{array}{l}
    1 + \beta^2 \left( \frac{2 - 3\mathfrak{p}}{\mathfrak{p}} \right) > 0.
  \end{array} \label{eqn-beta}
\end{equation}
To keep the notation simpler we set
\[ K (m, \beta, \mathfrak{p}) \assign m^2 \left( 1 + \beta^2 \left( \frac{2 -
   3\mathfrak{p}}{\mathfrak{p}} \right) \right) . \]
Thus, from \eqref{eqn8} we deduce that
\begin{equation}\label{eqn9}
    K (m, \beta, \mathfrak{p}) \| \bar{\rho}  \chi_{\varepsilon}
    \|_{L^{\mathfrak{p}} (\tilde{D}_1)}^{\mathfrak{p}} \leq
    M_{\theta}^{\mathfrak{p}} \left( \frac{4 (\mathfrak{p}- 1)}{\mathfrak{p}}
    + (2 m \beta + 1) \left( 1 - \frac{2}{\mathfrak{p}} \right) \right)
    \int_{\tilde{D}_1 \setminus \bar{D}_1} \rho^{\mathfrak{p}}  |
    \chi_{\varepsilon} |^{\mathfrak{p}},
\end{equation}
where $M_{\theta} > 0$ is the bound of $\theta$ and its derivatives up to
order $2$.\\

\noindent Further, since
\begin{equation}
  | \rho (x) | \leqslant e^{- m \beta \frac{l}{8}} \qquad \tmop{for} x \in
  \tilde{D}_1 \setminus \bar{D}_1, \label{eqn10}
\end{equation}
by substituting \eqref{eqn9} in \eqref{eqn0} we infer that
\begin{align}\label{eqn11}
    & \| f \cdot \varphi_{\varepsilon} (\cdot + x_1) - f \cdot \varphi_{1,
    \varepsilon} (\cdot + x_1) \|_{B_{p, p, \ell}^s}^{\mathfrak{p}} \nonumber \\
    & \quad \lesssim ~ \frac{e^{m\mathfrak{p} \beta} \| f \|_{L^{\infty}}}{K (m,
    \beta, \mathfrak{p})} M_{\theta}^{\mathfrak{p}} \left( \frac{4
    (\mathfrak{p}- 1)}{\mathfrak{p}} + (2 m \beta + 1) \left( 1 -
    \frac{2}{\mathfrak{p}} \right) \right) \int_{\tilde{D}_1 \setminus
    \bar{D}_1} \rho^{\mathfrak{p}}  | \chi_{\varepsilon} |^{\mathfrak{p}} \nonumber \\
    & \quad \lesssim_{m, \mathfrak{p}, M_{\theta}, \| f \|_{L^{\infty}}} ~~ 
    e^{m\mathfrak{p} \beta \left( 1 - \frac{l}{8} \right)} (\| X_{\varepsilon}
    - X_{1, \varepsilon} \|_{L^{\mathfrak{p}} (\tilde{D}_1 \setminus
    \bar{D}_1)}^{\mathfrak{p}} + \| \bar{\varphi}_{\varepsilon} -
    \bar{\varphi}_{1, \varepsilon} \|_{L^{\mathfrak{p}} (\tilde{D}_1 \setminus
    \bar{D}_1)}^{\mathfrak{p}}).
\end{align}
Thus, by applying $\mathbb{E}$ on both sides we get
\begin{align}\label{eqn12}
    & \mathbb{E} \left[\| f \cdot \varphi_{\varepsilon} (\cdot + x_1) - f \cdot
    \varphi_{1, \varepsilon} (\cdot + x_1) \|_{B_{p, p,
    \ell}^s}^{\mathfrak{p}} \right] \nonumber \\
	& \quad \lesssim_{m, \mathfrak{p}, M_{\theta}, \| f
    \|_{L^{\infty}}} ~~  e^{m\mathfrak{p} \beta \left( 1 - \frac{l}{8} \right)}
    \bigg( \mathbb{E} \left[\| X_{\varepsilon} - X_{1, \varepsilon} \|_{L^{\mathfrak{p}}
    (\tilde{D}_1 \setminus \bar{D}_1)}^{\mathfrak{p}} \right]  \nonumber \\
    &  \quad \quad   + \mathbb{E} \left[\| \bar{\varphi}_{\varepsilon}
    \|_{L^{\mathfrak{p}} (\tilde{D}_1 \setminus \bar{D}_1)}^{\mathfrak{p}} \right]
    +\mathbb{E} \left[\| \bar{\varphi}_{1, \varepsilon} \|_{L^{\mathfrak{p}}
    (\tilde{D}_1 \setminus \bar{D}_1)}^{\mathfrak{p}} \right] \bigg).
\end{align}

To estimate the term $\mathbb{E} [\| X_{\varepsilon} - X_{1, \varepsilon}
\|_{L^{\mathfrak{p}} (\tilde{D}_1 \setminus \bar{D}_1)}^{\mathfrak{p}}]$,
since $\tmop{supp} \mathfrak{a}_{\varepsilon} \subset B (0, \varepsilon),$ we
first infer that $\xi_{\varepsilon} = \xi_{1, \varepsilon}$ on $D_{1,
\varepsilon} \assign B \left( x_1, \frac{l}{2} - \varepsilon \right)$. By
using the representation from Lemma~\ref{lem-AY2002} we have that
\begin{equation}\label{eqn13}
    (X_{\varepsilon} - X_{1, \varepsilon}) (x) = \int_{\mathbb{R}^4} K  (x -
    z) \mathbbm{1}_{D_{1, \varepsilon}^c} (z) (\xi_{\varepsilon} (\mathd z) -
    \xi_{1, \varepsilon} (\mathd z) ) .
\end{equation}
Since, for $x \in \tilde{D}_1 \setminus \bar{D}_1$, we have $| x - z | >
\frac{l}{4} - \varepsilon \gg 1$ for $z \in D_{1, \varepsilon}^c$, thus  by
Lemma~\ref{lem-AY2002} (1) we obtain
\begin{align*}
     \mathbb{E} \left[((X_{\varepsilon} - X_{1, \varepsilon}) (x))^2 \right] & = \mathbb{E} \left[ \int_{\mathbb{R}^4} K  (x - z) \mathbbm{1}_{D_{1,
     \varepsilon}^c} (z) \xi_{\varepsilon} (\mathd z)  \int_{\mathbb{R}^4} K 
     (x - z_1) \mathbbm{1}_{D_{1, \varepsilon}^c} (z_1) \xi_{\varepsilon}
     (\mathd z_1) \right] \\
     & \quad -\mathbb{E} \left[ \int_{\mathbb{R}^4} K  (x - z) \mathbbm{1}_{D_{1,
     \varepsilon}^c} (z) \xi_{1, \varepsilon} (\mathd z)  \int_{\mathbb{R}^4}
     K  (x - z_1) \mathbbm{1}_{D_{1, \varepsilon}^c} (z_1) \xi_{\varepsilon}
     (\mathd z_1) \right]\\
     & \quad -\mathbb{E} \left[ \int_{\mathbb{R}^4} K  (x - z) \mathbbm{1}_{D_{1,
     \varepsilon}^c} (z) \xi_{\varepsilon} (\mathd z)  \int_{\mathbb{R}^4} K 
     (x - z_1) \mathbbm{1}_{D_{1, \varepsilon}^c} (z_1) \xi_{1, \varepsilon}
     (\mathd z_1) \right]\\
     & \quad +\mathbb{E} \left[ \int_{\mathbb{R}^4} K  (x - z) \mathbbm{1}_{D_{1,
     \varepsilon}^c} (z) \xi_{1, \varepsilon} (\mathd z)  \int_{\mathbb{R}^4}
     K  (x - z_1) \mathbbm{1}_{D_{1, \varepsilon}^c} (z_1) \xi_{1,
     \varepsilon} (\mathd z_1) \right]\\
     & \lesssim \int_{\mathbb{R}^8} C_1^2 {e^{- C_2 | x - z |}}  e^{- C_2 |
     x - z |} \mathbbm{1}_{D_{1, \varepsilon}^c} (z) \mathbbm{1}_{D_{1,
     \varepsilon}^c} (z_1)  \int_{\mathbb{R}^4} \mathfrak{a}_{\varepsilon} (z
     - z_2) \mathfrak{a}_{\varepsilon} (z_1 - z_2) \mathd z_2 \mathd z \mathd
     z_1 \\
     & \leq C_1^2 e^{- C_2 d (x, D_{1, \varepsilon}^c)}
     \int_{\mathbb{R}^4} e^{- C_2 | x - z_1 |} \mathd z_1 \\
     & \lesssim  e^{- C_2 d (x, D_{1, \varepsilon}^c)},
\end{align*}
which is finite and independent of $\varepsilon$. Here we have also employed the fact that
$\int_{\mathbb{R}^4} (\mathfrak{a}_{\varepsilon} \ast
\mathfrak{a}_{\varepsilon}) (z - z_1) \mathd z = 1$, which holds true because $\mathfrak{a}_{\varepsilon} \ast \mathfrak{a}_{\varepsilon}$ approximates
$\delta \ast \delta$, where $\delta$ represents the Dirac delta distribution.

\noindent Consequently, since $(X_{\varepsilon} - X_{1, \varepsilon}) (x)$ is Gaussian from \eqref{eqn13},
by hypercontractivity (see Theorem 3.50 in {\cite{Janson97}}) there exists a constant $C_{\mathfrak{p}} > 0$ such that, for every $x \in \tilde{D}_1 \setminus \bar{D}_1$,
\begin{equation}
  \mathbb{E} \left[| (X_{\varepsilon} - X_{1, \varepsilon}) (x) |^{\mathfrak{p}} \right]
  ~\leq~ C_{\mathfrak{p}} \left( \mathbb{E} \left[| (X_{\varepsilon} - X_{1, \varepsilon})
  (x) |^2 \right] \right)^{\frac{\mathfrak{p}}{2} } ~\lesssim~ C_{\mathfrak{p}} e^{-
  \frac{\mathfrak{p}}{2} C_2 d (x, D_{1, \varepsilon}^c)} \leqslant
  C_{\mathfrak{p}} . \label{eqn14}
\end{equation}
Furthermore, since 
$$ \mathbb{E} \left[\| X_{\varepsilon} - X_{1, \varepsilon}
\|_{L^{\mathfrak{p}} (\tilde{D}_1 \setminus \bar{D}_1)}^{\mathfrak{p}} \right] ~=~
\int_{\Omega} \int_{\tilde{D}_1 \setminus \bar{D}_1} | X_{\varepsilon} (x,
\omega) - X_{1, \varepsilon} (x, \omega) |^{\mathfrak{p}} \, \mathd x\, \mathbb{P}
(\mathd \omega),$$ 
the Fubini Theorem followed by \eqref{eqn14} yield
\begin{equation}
  \begin{array}{lllll}
    \mathbb{E} \left[\| X_{\varepsilon} - X_{1, \varepsilon} \|_{L^{\mathfrak{p}}
    (\tilde{D}_1 \setminus \bar{D}_1)}^{\mathfrak{p}} \right] & = & \int_{\tilde{D}_1
    \setminus \bar{D}_1} \mathbb{E} [| (X_{\varepsilon} - X_{1, \varepsilon})
    (x) |^{\mathfrak{p}}] \mathd x & \lesssim & C_{\mathfrak{p}} .
  \end{array} \label{eqn15}
\end{equation}

Finally, we assert that $\mathbb{E} [\| \bar{\varphi}_{\varepsilon}
\|_{L^{\mathfrak{p}}(\mathbb{R}^4)}^{\mathfrak{p}}] <
\infty$. This assertion trivially implies $\mathbb{E} [\| \bar{\varphi}_{\varepsilon}
\|_{L^{\mathfrak{p}} (\tilde{D}_1 \setminus \bar{D}_1)}^{\mathfrak{p}}] <
\infty$ in \eqref{eqn12}. 
We start the proof of this claim by recalling from Subsection
\ref{subsec-coupling-notation} that $\alpha \bar{\varphi}_{\varepsilon}
\leqslant 0$ and $\bar{\varphi}_{\varepsilon}$ is a unique solution to
\begin{equation}
  \begin{array}{l}
    \mathcal{L} \bar{\varphi}_{\varepsilon} + \alpha \exp (\alpha
    \bar{\varphi}_{\varepsilon}) \eta_{\varepsilon} = 0.
  \end{array} \label{moll-good-SPDE}
\end{equation}
By testing \eqref{moll-good-SPDE} with $\rho^{\mathfrak{p}} |
\bar{\varphi}_{\varepsilon} |^{\mathfrak{p}- 2}  \bar{\varphi}_{\varepsilon}$
and integrating it on $\mathbb{R}^4$ we obtain
\begin{equation}
    \int \rho^{\mathfrak{p}} | \bar{\varphi}_{\varepsilon} |^{\mathfrak{p}- 2}
    \bar{\varphi}_{\varepsilon} \mathcal{L} \bar{\varphi}_{\varepsilon} +
    \int \alpha \rho^{\mathfrak{p}} | \bar{\varphi}_{\varepsilon}
    |^{\mathfrak{p}- 2}  \bar{\varphi}_{\varepsilon} \exp (\alpha
    \bar{\varphi}_{\varepsilon}) \eta_{\varepsilon} = 0. \label{eqn16}
\end{equation}
Since from \eqref{eqn2} and \eqref{eqn4}
\begin{align*}
    \int \rho^{\mathfrak{p}} | \bar{\varphi}_{\varepsilon} |^{\mathfrak{p}-
     2}  \bar{\varphi}_{\varepsilon} \mathcal{L} \bar{\varphi}_{\varepsilon} &
     =  m^2 \int \rho^{\mathfrak{p}} | \bar{\varphi}_{\varepsilon}
     |^{\mathfrak{p}} + \int \rho^{\mathfrak{p}- 1} |
     \bar{\varphi}_{\varepsilon} |^{\mathfrak{p}- 2}
     \bar{\varphi}_{\varepsilon} (- \Delta (\rho
     \bar{\varphi}_{\varepsilon})) \\
     &  \quad  + \left( 1 - \frac{2}{\mathfrak{p}} \right) \int
     \rho^{\mathfrak{p}- 1} | \bar{\varphi}_{\varepsilon} |^{\mathfrak{p}}
     \Delta \rho - \frac{2 (\mathfrak{p}- 1)}{\mathfrak{p}} \int |
     \bar{\varphi}_{\varepsilon} |^{\mathfrak{p}} \rho^{\mathfrak{p}- 2} |
     \nabla \rho |^2, 
\end{align*}
where $\int \rho^{\mathfrak{p}- 1} | \bar{\varphi}_{\varepsilon}
|^{\mathfrak{p}- 2} \bar{\varphi}_{\varepsilon} (- \Delta (\rho
\bar{\varphi}_{\varepsilon})) \geq 0$, \eqref{eqn16} gives
\begin{align}
    & m^2 \int \rho^{\mathfrak{p}} | \bar{\varphi}_{\varepsilon}
    |^{\mathfrak{p}} + \int \alpha \rho^{\mathfrak{p}} |
    \bar{\varphi}_{\varepsilon} |^{\mathfrak{p}- 2}
    \bar{\varphi}_{\varepsilon} \exp (\alpha \bar{\varphi}_{\varepsilon})
    \eta_{\varepsilon} \nonumber \\
    & \qquad \quad = ~ \frac{2 (\mathfrak{p}- 1)}{\mathfrak{p}} \int |
    \bar{\varphi}_{\varepsilon} |^{\mathfrak{p}} \rho^{\mathfrak{p}- 2} |
    \nabla \rho |^2 - \left( 1 - \frac{2}{\mathfrak{p}} \right) \int
    \rho^{\mathfrak{p}- 1} | \bar{\varphi}_{\varepsilon} |^{\mathfrak{p}}
    \Delta \rho .
  \label{eqn17}
\end{align}
Since $\eta_{\varepsilon} \rho^{\mathfrak{p}}$ is a positive distribution, the second l.h.s. term in \eqref{eqn17} can be estimated as
\begin{equation*}
    \left| \int \alpha \rho^{\mathfrak{p}} | \bar{\varphi}_{\varepsilon}
     |^{\mathfrak{p}- 2}  \bar{\varphi}_{\varepsilon} \exp (\alpha
     \bar{\varphi}_{\varepsilon}) \eta_{\varepsilon} \right| ~\leq~ \| ~ |
     \bar{\varphi}_{\varepsilon} |^{\mathfrak{p}- 2} \alpha
     \bar{\varphi}_{\varepsilon} \exp (\alpha \bar{\varphi}_{\varepsilon})
     \mathbb{I} (\alpha \bar{\varphi}_{\varepsilon}) \|_{C (\mathbb{R}^4)}
     \left[ \int \eta_{\varepsilon} \rho^{\mathfrak{p}} \right],
\end{equation*}
where $\mathbb{I}: \mathbb{R} \rightarrow \mathbb{R}_+$ is a smooth function supported on $(- \infty, 1)$. Note that $\mathbb{I} (\alpha
\bar{\varphi}_{\varepsilon}) = 1$,  since $\alpha \bar{\varphi}_{\varepsilon}
\leqslant 0$. But, for each $x \in \mathbb{R}^4$,
\begin{align*}
    | ~ | \bar{\varphi}_{\varepsilon} |^{\mathfrak{p}- 2} \alpha
     \bar{\varphi}_{\varepsilon} \exp (\alpha \bar{\varphi}_{\varepsilon})
     \mathbb{I} (\alpha \bar{\varphi}_{\varepsilon}) |  & =~  | \alpha |^{2
     -\mathfrak{p}} | | \alpha \bar{\varphi}_{\varepsilon} |^{\mathfrak{p}- 2}
     \alpha \bar{\varphi}_{\varepsilon} \exp (\alpha
     \bar{\varphi}_{\varepsilon}) \mathbb{I} (\alpha
     \bar{\varphi}_{\varepsilon}) | \\
     & \leq  ~ | \alpha |^{2 -\mathfrak{p}} \sup_{z \in \mathbb{R}^4} [z
     | z |^{\mathfrak{p}- 2} \exp (z) \mathbb{I} (z)]\\
     & \leq  ~ C | \alpha |^{2 -\mathfrak{p}},
\end{align*}
for some $C > 0$, where the r.h.s is independent of $\varepsilon$ and $x$. By
substituting the above estimate into \eqref{eqn17} we obtain
\begin{align}\label{eqn18}
    m^2 \int \rho^{\mathfrak{p}} | \bar{\varphi}_{\varepsilon}
    |^{\mathfrak{p}} & ~ \leq ~ \frac{2 (\mathfrak{p}- 1)}{\mathfrak{p}}
    \int | \bar{\varphi}_{\varepsilon} |^{\mathfrak{p}} \rho^{\mathfrak{p}}
    \frac{| \nabla \rho |^2}{\rho^2} - \left( 1 - \frac{2}{\mathfrak{p}}
    \right) \int \rho^{\mathfrak{p}} | \bar{\varphi}_{\varepsilon}
    |^{\mathfrak{p}} \frac{\Delta \rho}{\rho} \nonumber\\
    & \qquad + C | \alpha |^{2 -\mathfrak{p}} \int \eta_{\varepsilon}
    \rho^{\mathfrak{p}} .
\end{align}
Now since $\nabla \rho (x) = - m \beta \frac{x - x_1}{| x - x_1 |} \rho (x)$
for $x \in \mathbb{R}^4 \setminus \{ x_1 \}$ and $\Delta \rho = \rho m^2
\beta^2$, we can choose $\beta > 0$ \ such that
\begin{equation}
    \frac{2 (\mathfrak{p}- 1)}{\mathfrak{p}} \frac{| \nabla \rho |^2}{\rho^2}
    - \left( 1 - \frac{2}{\mathfrak{p}} \right) \frac{\Delta \rho}{\rho} ~=~ m^2
    \beta^2 ~\leqslant~ \frac{m^2}{2} .
    \label{beta-2}
\end{equation}
Consequently, with $\beta$ such that \eqref{eqn-beta} and \eqref{beta-2} hold
true, from \eqref{eqn18} we deduce that
\begin{equation}
    \frac{m^2}{2} \| \rho \bar{\varphi}_{\varepsilon}
    \|_{L^{\mathfrak{p}}}^{\mathfrak{p}} ~\leqslant~ C | \alpha |^{2
    -\mathfrak{p}} \int \eta_{\varepsilon} \rho^{\mathfrak{p}} .
  \label{eqn19} 
\end{equation}
Thus, $\mathbb{E} \left[\| \rho \bar{\varphi}_{\varepsilon}
\|_{L^{\mathfrak{p}}}^{\mathfrak{p}} \right] < \infty$ and the bound is uniform in $\varepsilon$ because
\begin{equation}\label{eqn20} 
  \mathbb{E} \left[ \int \eta_{\varepsilon} \rho^{\mathfrak{p}} \right] = \int
  e^{- m\mathfrak{p} \beta | x - x_1 |} \mathd x < \infty . 
\end{equation}
Similarly we can show that $\mathbb{E} [\| \rho \bar{\varphi}_{1, \varepsilon}
\|_{L^{\mathfrak{p}}}^{\mathfrak{p}}] < \infty$ uniformly in $\varepsilon$.\\

\noindent Hence, substituting \eqref{eqn15} together with \eqref{eqn19} and 
 the uniform boundedness of $\mathbb{E} [\| \rho \bar{\varphi}_{\varepsilon}
\|_{L^{\mathfrak{p}}}^{\mathfrak{p}}]$ and $\mathbb{E} [\| \rho
\bar{\varphi}_{1, \varepsilon} \|_{L^{\mathfrak{p}}}^{\mathfrak{p}}]$
from~\eqref{eqn12}, for $\beta$ satisfying \eqref{eqn-beta} and
\eqref{beta-2}, we have
\begin{equation}
  \begin{array}{lll}
    \mathbb{E} \left[\| f \cdot \varphi_{\varepsilon} (\cdot + x_1) - f \cdot
    \varphi_{1, \varepsilon} (\cdot + x_1) \|_{B_{p, p,
    \ell}^s}^{\mathfrak{p}} \right] & \lesssim_{m, \mathfrak{p}, M_{\theta}, | \alpha
    |, \| f \|_{L^{\infty}}} & e^{- m \beta \frac{l}{8}},
  \end{array} \label{eqn21}
\end{equation}
which is independent of $\varepsilon$ and $x_1$. Here we have also used \
$e^{m \beta \left( 1 - \frac{l}{8} \right)} \simeq_{m, \mathfrak{p}} e^{- m
\beta \frac{l}{8}}$. Hence we get \eqref{formal-decay} and due to inequality
\eqref{coupling-cov-est1} the proof of Theorem \ref{thm-main} is complete.

\section{The Malliavin calculus approach}\label{sec-malliavin}
\qquad In this section our aim is to present the proof of
Theorem~\ref{thm-main} via the approach based on Malliavin calculus. The proof will start by considering an approximation $\varphi_{\varepsilon, R}$ useful
to be able to apply easily the Malliavin calculus, see
eqns.~\eqref{SPDE-Kregapprox0} and~\eqref{SPDE-e-R-approx}. The solution
theory to~\eqref{SPDE-Kregapprox0} is closely related to Lemmata 30 and 31
of~{\cite{AVG2021}} and proved in Proposition~\ref{prop-reg soln of
varphi-bar} and Lemma~\ref{lemma_exponentialuniqueness} below. The Malliavin
calculus enters in estimating $\tmop{Cov} (\varphi_{\varepsilon, R} (x_1),
\varphi_{\varepsilon, R} (x_2))$ in terms of the Malliavin derivative of
$\varphi_{\varepsilon, R}$ which we denote by $D_z \varphi_{\varepsilon, R}$,
see eqs.~\eqref{covExpression}, \eqref{rhs2Est} and \eqref{rhs1Est}. The
existence of $D_z \varphi_{\varepsilon, R}$ and the linear elliptic SPDE it
satisfies are established in Theorem~\ref{thm-MallDiff} thanks to a
preliminary abstract result from~{\cite{Tindel98}} which we state as
Theorem~\ref{thm-MallavinDiffAbstract}. Finally the Feynman--Kac formula and
some estimates from Malliavin calculus, for example \eqref{corB.6}, help us to
finish the proof.

\subsection{Preliminaries}\label{subsec-malliavin-notation}

\qquad Before moving on, let us first recall the tools from Malliavin calculus that
we will need. Most of the definitions and preliminary results here are taken
from Chapter~1 of Nualart's book~{\cite{Nualart2006}}. Let $H$ be a separable
Hilbert space and $W = \{ W (h), h \in H \}$ an isonormal Gaussian process
defined on a complete probability space $(\Omega, \mathfrak{F}, \mathbb{P})$.
Let $\mathcal{E}$ be the $\sigma$-field generated by the random variables $\{
W (h), h \in H \}$. Since $\mathcal{E} \subseteq \mathfrak{F}$, note that when
we write $(\Omega, \mathcal{E}, \mathbb{P})$ we mean that $\mathbb{P}$ is the
restriction of the probability measure defined on $\mathfrak{F}$ to
$\mathcal{E}$.

For each $n \geqslant 0$ by $H_n (x)$ we denote the well known $n$th Hermite
polynomial and by $\mathcal{H}_n,$ the Wiener chaos of order $n$, that is, the
closed linear subspace of $L^2 (\Omega, \mathcal{E}, \mathbb{P})$ generated by
the random variables $\{ H_n (W (h)), h \in H, \| h \|_{H} = 1 \}$ whenever $n
\geqslant 1$, and the set of constants for $n = 0$. One of the important
results in the Malliavin calculus is the Wiener chaos decomposition of $L^2
(\Omega, \mathcal{E}, \mathbb{P})$ into its projections in the spaces
$\mathcal{H}_n$, i.e.,
\[ L^2 (\Omega, \mathcal{E}, \mathbb{P}) = \oplus_{n \geqslant 0}
   \mathcal{H}_n . \]
In particular for any $F \in L^2 (\Omega, \mathcal{E}, \mathbb{P})$, we have $F = \sum_{n = 0}^{\infty} J_n F$ where $J_n F$ denotes the projection of $F$ into $\mathcal{H}_n$. We will restrict our discussion of this section to $L^2 (\Omega, \mathcal{E}, \mathbb{P})$ and to shorten the notation we will denote it by $L^2 (\Omega)$.

The Malliavin derivative operator $D$ maps the domain $\mathbb{D}^{1, 2}
\subseteq L^2 (\Omega) $ to the space of $H$-valued random variables $L^2
(\Omega ; H) .$ Note that $F \in \mathbb{D}^{1, 2}$ if and only if $\sum_{n =
1}^{\infty} n \| J_n F \|_{L^2 (\Omega)}^2 < \infty$. Moreover, in this
setting for all $n \geqslant 1$, we have
\begin{equation}\label{DandJrelation}
  D (J_n F) = J_{n - 1} (D F) . 
\end{equation}
The divergence operator $\delta : \tmop{Dom} \delta \subseteq L^2 (\Omega ; H)
\rightarrow L^2 (\Omega)$ is defined as the adjoint of the derivative operator
$D$. We will work in the special case of $H = L^2 (T, \mathcal{B}, \tau)$,
where $(T, \mathcal{B})$ is a measurable space and $\tau$ is a $\sigma$-finite
atom-less measure on $(T, \mathcal{B})$. Also, we will identify $L^2 (\Omega ;
L^2 (T))$ with $L^2 (T \times \Omega)$ which is the set of square integrable
stochastic processes. Thus, for $F \in \mathbb{D}^{1, 2}$, $D F \in L^2 (T
\times \Omega)$ and we write $D_t F = D F (t), ~ \forall t \in T$.
By $\mathbb{D}^{1, 2} (L^2 (T))$ we denote the set of stochastic processes $u
\in L^2 (T \times \Omega)$ such that $u (t) \in \mathbb{D}^{1, 2}$ for almost all $t \in T$ and there exists a measurable version of the two parameter process $\{ D_s u (t) \}_{s, t \in T} \subset L^2 (\Omega)$ satisfying
\[ \mathbb{E} \left[ \int_T \int_T (D_s u (t))^2 \, \tau(\mathd s)\,  \tau (\mathd t)
   \right] < \infty . \]
In the Malliavin calculus literature, the space $\mathbb{D}^{1, 2} (L^2 (T))$
is generally denoted by $\mathbb{L}^{1, 2}$. Note that $\mathbb{L}^{1, 2}$ is
a subset of $\tmop{Dom} \delta$ and isomorphic to $L^2 (T ; \mathbb{D}^{1, 2})
.$ Then, see (1.54) of {\cite{Nualart2006}}, for $u, v \in \mathbb{L}^{1, 2}$
we have
\begin{equation}
  \mathbb{E} [\delta (u) \delta (v)] = \int_T \mathbb{E} [u (t) v (t)] \, \tau
  (\mathd t) + \int_T \int_T \mathbb{E} [D_s u (t) D_s v (t)] \, \tau (\mathd s)\,  \tau
  (\mathd t) . \label{LandDrelation}
\end{equation}

Let $\{ P_t, t \geqslant 0 \}$ be the one parameter Ornstein-Uhlenbeck
semigroup of contraction operators in $L^2 (\Omega)$ and by $L : L^2 (\Omega)
\ni F \rightarrow \sum_{n = 0}^{\infty} - n J_n F \in L^2 (\Omega)$ we denotes
its infinitesimal generator with domain
\[ \tmop{Dom} L = \left\{ F \in L^2 (\Omega) : \sum_{n = 0}^{\infty} n^2 \|
   J_n F \|_{L^2 (\Omega)} < \infty \right\} . \]
From Proposition~1.4.3 of {\cite{Nualart2006}} we know that, for $F \in L^2
(\Omega),$ $F \in \tmop{Dom} L$ if and only if $F \in \mathbb{D}^{1, 2}$ and
$D F \in \tmop{Dom} \delta$. In this case we have $\delta D F = - L F.$

With the above notation, equality (90) in {\cite{FG2019}} gives the following commutation property
\begin{equation}
  D (I - L)^{- 1} F = (2 I - L)^{- 1} D F, \quad \forall F \in L^2 (\Omega),
  \label{FG2019_90}
\end{equation}
and the proof of Lemma~B.1 in {\cite{FG2019}} give the following first order
expansion
\begin{equation}
  F -\mathbb{E} [F] = \delta (I - L)^{- 1} D F, \qquad \forall F \in
  \mathbb{D}^{1, 2} . \label{lemB.1}
\end{equation}
To proceed with our analysis, let us fix the $\sigma$-finite measure space
$(T, \mathcal{B}, \tau)$ as $(\mathbb{R}^4, \mathcal{B} (\mathbb{R}^4), \mathd
x)$ where $\mathcal{B} (\mathbb{R}^4)$ denotes the Borel $\sigma$-field on $\mathbb{R}^4$ and $\mathd x$ stands for the Lebesgue measure.

\subsection{Proof of Theorem \ref{thm-main}}\label{subsec-malliavin-proof}

\quad We recall that $\xi$ is a given space white noise on $\mathbb{R}^4$. Thus, the
isonormal Gaussian process we consider here is $W (h) = \langle \xi, h
\rangle, h \in L^2_{\ell} (\mathbb{R}^4)$, indexed by the Hilbert space
$L^2_{\ell} (\mathbb{R}^4) .$ We will be working under the framework of
Malliavin calculus associated to white noise $\xi$ on $\mathbb{R}^4$. To
setup, let $\Omega = B^{- 2 - \kappa}_{\infty, \infty, \ell} (\mathbb{R}^4)$
and let $\mathbb{P}$ be the law of $\xi$ on $\Omega .$

It turns out that the following approximation of the equation
\eqref{goodSPDE}, instead of \eqref{good-approx-SPDE}, is more suitable to work with the above mentioned tools from Malliavin calculus
\begin{equation}  \label{SPDE-Kregapprox0}
      \mathcal{L} \bar{\varphi}_{\varepsilon, R} + \alpha K_R (\exp (\alpha
      \bar{\varphi}_{\varepsilon, R}) \exp (\alpha X_{\varepsilon} -
      C_{\varepsilon})) = 0,
\end{equation}
where $K_R : (0, \infty) \rightarrow (0, \infty)$ is a smooth function which
is equal to $x$ if $x \in (0, R - 1]$, equal to $R$ if $x \geqslant R$ and
$K_R$ is increasing for $x \in (R - 1, R)$. Since the proof presented here of the solution theory to equation \eqref{SPDE-Kregapprox0} is closely related to Lemmata 30 and 31 of~{\cite{AVG2021}}, the results about the existence of a unique solution $\bar{\varphi}_{\varepsilon, R}$ to \eqref{SPDE-Kregapprox0} are postponed to Proposition~\ref{prop-reg soln of varphi-bar} and Lemma~\ref{lemma_exponentialuniqueness} in Appendix \ref{sec-appendix}. Moreover, it is straightforward to see that $\bar{\varphi}_{\varepsilon, R}
\rightarrow \bar{\varphi}_{\varepsilon}$ as $R \rightarrow \infty$, where
$\bar{\varphi}_{\varepsilon}$ is the unique solution to the equation
\eqref{good-approx-SPDE}.

Further recall, from \eqref{def-eta-eps}, that we denote the expression $\exp (\alpha X_{\varepsilon} -
C_{\varepsilon})$ by $\eta_{\varepsilon}$. Let us define the following random field
\begin{equation}
  (\mathcal{G}_{\varepsilon} \ast \xi) (x) \assign \int_{\mathbb{R}^4}
  (\mathfrak{a}_{\varepsilon} \ast \mathcal{G}) (x - y) \, \xi (\mathd y), \qquad x \in
  \mathbb{R}^4, \label{defn-varphi-0}
\end{equation}
where $\mathcal{G}$ is the Green function associated with the operator $(-
\Delta + m^2)^{- 1}$ and $\mathcal{G}_{\varepsilon} \assign
\mathfrak{a}_{\varepsilon} \ast \mathcal{G}$. It can be shown that
$\mathcal{G}_{\varepsilon} \ast \xi$ is a smooth Gaussian process, see Theorem 5.1 of {\cite{mourratGlobalWellposednessDynamic2017}}. 

By setting $\varphi_{\varepsilon, R} = \bar{\varphi}_{\varepsilon, R} +
X_{\varepsilon}$, from \eqref{SPDE-Kregapprox0} we get that
$\varphi_{\varepsilon, R}$ uniquely solves the following equation
\begin{equation}\label{SPDE-e-R-approx}
    \mathcal{L} \varphi_{\varepsilon, R} + \alpha K_R (\exp (\alpha
    \varphi_{\varepsilon, R} - C_{\varepsilon})) = \xi_{\varepsilon},
\end{equation}
which is equivalent to say that, for $x \in \mathbb{R}^4$ and $\omega \in \Omega$,
\begin{equation}\label{int-form-e-R-approx}
    \varphi_{\varepsilon, R} (x, \omega) + \alpha \int_{\mathbb{R}^4}
    \mathcal{G} (x - y) K_R (\exp (\alpha \varphi_{\varepsilon, R} (y, \omega)
    - C_{\varepsilon})) \, \mathd y = (\mathcal{G}_{\varepsilon} \ast \xi) (x) .
\end{equation}
To shorten the notation we will write \begin{equation}
    \int_{\mathbb{R}^4} \mathcal{G} (x - y) K_R (\exp (\alpha
     \varphi_{\varepsilon, R} (y, \omega) - C_{\varepsilon})) \, \mathd y =
     (\mathcal{G} \ast K_R (\exp (\alpha \varphi_{\varepsilon, R} (\cdot,
     \omega) - C_{\varepsilon}))) (x) .    
\end{equation}

Since one can write the term $\tmop{Cov} (F_1 (f \cdot \varphi_{\varepsilon,
R} (\cdot + x_1)), F_2 (f \cdot \varphi_{\varepsilon, R} (\cdot + x_2)))$,
that we want to estimate, in terms of $D_z \varphi_{\varepsilon, R}$, see
\eqref{covExpression} for precise expression, we aim next to find the equation
for $D_z \varphi_{\varepsilon, R}$. This we achieve in Theorem
\ref{thm-MallDiff} whose proof is based on the following abstract result which
is stated as Theorem 2.5 in~{\cite{Tindel98}}.
\begin{theorem}
  \label{thm-MallavinDiffAbstract}Let $(\Omega, \mathbb{P})$ be a complete probability space on which $\xi$ is a canonical process. Further, assume that  $H$ is continuously embedded in $\Omega$ and let us denote
  this embedding by $i$. Let $F \in L^2 (\Omega)$. Then $F \in \mathbb{D}^{1,
  2}$ iff the following conditions are satisfied.
  \begin{enumerate}
    \item For all $h \in H$, there exists a version $\tilde{F}_h$ of $F$ such
    that, for every $\omega \in \Omega$, the mapping $\mathbb{R} \ni t \mapsto
    \tilde{F}_h [\omega + t i (h)]$ is absolutely continuous.
    
    \item There exists $\varsigma \in L^2 (\Omega ; H)$ such that, for all \
    $h \in H,$
    \begin{equation}
        \lim_{t \rightarrow 0} \frac{1}{t} \{ F [\omega + t i (h)] - F
         (\omega) \} = \langle \varsigma (\omega), h \rangle, \qquad
         \mathbb{P} \textrm{-a.s.}
    \end{equation}
  \end{enumerate}
\end{theorem}
From the proof of Theorem \ref{thm-MallDiff} it can be observed that we apply
Theorem \ref{thm-MallavinDiffAbstract}, for each $x \in \mathbb{R}^4,$
$\varepsilon$ and $R$ on $F$ with $H \assign L_{\ell}^2 (\mathbb{R}^4)$ where
\begin{equation} \label{defn-F}
    F (\omega) \assign \varphi_{\varepsilon, R} (x, \omega), \quad \omega \in \Omega .
\end{equation}
Since most of the results of this section are independent of $\varepsilon, R$
and $x$ or for fixed $\varepsilon, R$ and $x$, unless otherwise stated we will
not write the explicit dependence of functions defined here on $\varepsilon,
R$ and $x$.

To study the required properties of $F$, which allow us to apply Theorem
\ref{thm-MallavinDiffAbstract}, we write \eqref{int-form-e-R-approx} in the
functional form as, for $\omega \in \Omega$,
\begin{equation}
  \begin{array}{l}
    \mathcal{T} (\varphi_{\varepsilon, R} (\cdot, \omega)) =
    (\mathcal{G}_{\varepsilon} \ast \xi) (\cdot) .
  \end{array} \label{rel-varphi-0}
\end{equation}
Here $\mathcal{T}$ is defined as
\begin{equation}
  \begin{array}{l}
    \mathcal{T}: \mathcal{B} \ni w \rightarrow w + \alpha \mathcal{G} \ast K_R
    (\exp (\alpha w - C_{\varepsilon})) \in \mathcal{B} \assign
    C_{\ell}^0 (\mathbb{R}^4) .
  \end{array} \label{T-map}
\end{equation}
Note that, because of the convolution, the map $\mathcal{T}$ is well-defined.
Moreover, by definition of the map $\mathcal{T}$, \eqref{defn-F} can be
understood as, for each $\omega \in \Omega$,
\begin{equation}
  \begin{array}{l}
    F (\omega) = (\mathcal{T}^{- 1} (\mathcal{G}_{\varepsilon} \ast \xi)) (x)
    . \label{eq:F-and-T}
  \end{array}
\end{equation}
Thus, because of \eqref{eq:F-and-T}, in order to study $F$ we first show in
Lemma~\ref{lem-T-bijectivity} that $\mathcal{T}^{- 1}$ exists, i.e., prove the
bijectivity of the map $\mathcal{T}$. This is precisely our next result.
Before this we prove an auxiliary result as follows.
\begin{lemma}
  Let $\mathfrak{v} \in L_{\ell}^2 (\mathbb{R}^4) $ and $\mathfrak{u} \in
  H_{\ell}^2 (\mathbb{R}^4)$ be a unique weak solution to $(- \Delta + m^2)
  \mathfrak{u}=\mathfrak{v}$. Then
  \begin{equation}\label{est-poincare}
      \langle \nabla (\mathcal{G} \ast \mathfrak{v}), (\mathcal{G} \ast
      \mathfrak{v}) \nabla r_{\ell}^2 \rangle + m^2 \| \mathcal{G} \ast
      \mathfrak{v} \|_{L^2_{\ell} (\mathbb{R}^4)}^2 \leqslant \langle
      \mathcal{G} \ast \mathfrak{v}, \mathfrak{v} \rangle_{\ell},
  \end{equation}
  where $\langle \cdot, \cdot \rangle$ and $\langle \cdot, \cdot
  \rangle_{\ell}$, respectively, denote the standard inner product in $L^2
  (\mathbb{R}^4)$ and $L^2_{\ell} (\mathbb{R}^4)$. 
\end{lemma}

\begin{proof}
  Let $\mathfrak{u}=\mathcal{G} \ast \mathfrak{v} \in H_{\ell}^2
  (\mathbb{R}^4)$ be a unique weak solution to $(- \Delta + m^2)
  \mathfrak{u}=\mathfrak{v}$ for given $\mathfrak{v} \in L_{\ell}^2
  (\mathbb{R}^4)$.
  Multiplying on both sides of $(- \Delta + m^2) \mathfrak{u}=\mathfrak{v}$ by   $r_{\ell}^2 \mathfrak{u}$ give
  \begin{equation}
      \langle (- \Delta + m^2) \mathfrak{u}, \mathfrak{u} \rangle_{\ell} = \langle \mathfrak{u}, \mathfrak{v} \rangle_{\ell} .
  \end{equation}
  Integration by parts yield,
  \[ \langle \nabla \mathfrak{u}, \mathfrak{u} \nabla r_{\ell}^2 \rangle +
     m^2 \| \mathfrak{u} \|_{L^2_{\ell} (\mathbb{R}^4)}^2 \leqslant \langle
     \mathfrak{u}, \mathfrak{v} \rangle_{\ell} . \]
  By substituting $\mathfrak{u}=\mathcal{G} \ast \mathfrak{v}$, above gives
  the conclusion. 
\end{proof}

To avoid complexity in notation we set $G (w) \assign \alpha K_R (\exp
(\alpha w - C_{\varepsilon})), w \in \mathcal{B}$. Then, $G$ is non-negative, bounded, smooth and non-decreasing.
\begin{lemma}\label{lem-T-bijectivity}
The map $\mathcal{T}$ is bijective from $\mathcal{B}$ onto $\mathcal{B}$. 
\end{lemma}
\begin{proof}
  Let us first show that $\mathcal{T}$ is one-one. In particular, we show that
  for small enough $\lambda > 0$ if $u, v \in \mathcal{B} \subset L_{\lambda,
  \ell'}^2 (\mathbb{R}^4)$ for $\ell \leqslant \ell'$ such that $\mathcal{T}u
  =\mathcal{T}v$, then $u = v$.
  
  Since $\mathcal{T}u =\mathcal{T}v,$ we have
  \begin{equation} \label{Tmap-1}
      u - v + [\mathcal{G} \ast G (u) -\mathcal{G} \ast G (v)] = 0.    
  \end{equation}
  Multiply this by $r_{\lambda, \ell'} (G (u) - G (v))$ and integrate on
  $\mathbb{R}^4$ to get
  \begin{equation}
      \langle u - v, G (u) - G (v) \rangle_{\lambda, \ell'} +
        \langle \mathcal{G} \ast (G (u) - G (v)), G (u) - G (v)
    \rangle_{\lambda, \ell'} = 0, \label{uv-eq1}
  \end{equation}
  where $\langle a, b \rangle_{\lambda, \ell'} \assign \int a (x) b (x) (1 + \lambda | x |^2)^{- \ell'} d x.$
  Consequently, since $G$ is non-decreasing and $\langle u - v, G (u) - G (v)
  \rangle_{\lambda, \ell'} \geqslant 0$, from \eqref{uv-eq1} we get
  \begin{equation}
    \langle \mathcal{G} \ast (G (u) - G (v)), G (u) - G (v) \rangle_{\lambda
    \comma \ell'} \leqslant 0. \label{Tmap-1-est1}
  \end{equation}
  But by substituting $G (u) - G (v)$ in place of $\mathfrak{v}$ in   \eqref{est-poincare} we obtain
  \begin{align}
      & \langle \nabla (\mathcal{G} \ast (G (u) - G (v))), (\mathcal{G} \ast
      (G (u) - G (v))) \nabla r_{\lambda, \ell'}^2 \rangle
      + m^2 \| \mathcal{G} \ast (G (u) - G (v)) \|_{L^2_{\lambda, \ell'}
      }^2 \nonumber \\
      & \quad \leq \langle \mathcal{G} \ast (G (u) - G (v)), (G (u) - G (v)) \rangle_{\lambda, \ell'} .
    \label{est-poincare-1}
  \end{align}
  So, using \eqref{Tmap-1-est1} in above yield
  \begin{equation}
      \langle \nabla (\mathcal{G} \ast (G (u) - G (v))), (\mathcal{G} \ast (G
      (u) - G (v))) \nabla r_{\lambda, \ell'}^2 \rangle + m^2 \| \mathcal{G}
      \ast (G (u) - G (v)) \|_{L^2_{\lambda, \ell'}}^2
      \leqslant 0.
    \label{uv-est2}
  \end{equation}
  But due to the integration by parts we have
  \begin{align*}
        & \langle \nabla (\mathcal{G} \ast (G (u) - G (v))), (\mathcal{G} \ast
       (G (u) - G (v))) \nabla r_{\lambda, \ell'}^2 \rangle\\
       & =  - \langle \nabla (\mathcal{G} \ast (G (u) - G (v))), (\mathcal{G}
       \ast (G (u) - G (v))) \nabla r_{\lambda, \ell'}^2 \rangle - \int (\mathcal{G} \ast (G (u) - G (v)))^2 \Delta r_{\lambda,
       \ell'}^2 \mathd x.   
  \end{align*}
  This gives
  \begin{align}
      & 2 \langle \nabla (\mathcal{G} \ast (G (u) - G (v))), (\mathcal{G}
      \ast (G (u) - G (v))) \nabla r_{\lambda, \ell'}^2 \rangle  =- \int (\mathcal{G} \ast (G (u) - G (v)))^2 \Delta r_{\lambda,
      \ell'}^2 \mathd x,
   \label{uv-est3}
  \end{align}
  where $r_{\lambda, \ell'}^2 (x) = (1 + \lambda | x |^2)^{- \ell'}$ \ and
  $\nabla r_{\ell'}^2 (x) = - 2 \lambda \ell' (1 + \lambda | x |^2)^{- (\ell'
  + 1)} x$ and
  \begin{equation}
      \Delta r_{\ell'}^2 (x)  =  - 4 \lambda \ell (1 + \lambda | x |^2)^{-
      (\ell + 1)} + 4 \lambda^2 \ell (\ell + 1) | x |^2 (1 + \lambda | x
      |^2)^{- (\ell + 2)} .
    \label{r-est1}
  \end{equation}
  Hence, substituting \eqref{r-est1} into \eqref{uv-est3} give
  \begin{align}
      & 2 \langle \nabla (\mathcal{G} \ast (G (u) - G (v))), (\mathcal{G}
      \ast (G (u) - G (v))) \nabla r_{\lambda, \ell'}^2 \rangle \nonumber \\
      & \qquad = - 4 \lambda^2 \ell' (\ell' + 1) \int (\mathcal{G} \ast (G (u) - G
      (v)))^2  | x |^2 (1 + \lambda | x |^2)^{- (\ell' + 2)} \mathd x \nonumber \\
      & \qquad \qquad + 4 \lambda \ell' \int (\mathcal{G} \ast (G (u) - G (v)))^2  (1 +
      \lambda | x |^2)^{- (\ell' + 1)} \mathd x.
    \label{uv-est4}
  \end{align}
  Consequently, using \eqref{uv-est4} into \eqref{uv-est2} provides
  \begin{equation}
      ( m^2 {- 2 \lambda \ell'}^2 ) \| \mathcal{G} \ast (G (u) - G
      (v)) \|_{L^2_{\lambda, \ell'}}^2 \leqslant 0.
    \label{Tmap-2}
  \end{equation}
  By taking sufficiently small $\lambda$, using \eqref{Tmap-1} together with
  \eqref{Tmap-2} we get $\| u - v \|_{L^2_{\lambda, \ell'}}^2
  \leqslant 0$. This implies $u = v$ in $L^2_{\lambda, \ell'} (\mathbb{R}^4)$
  and hence the map $\mathcal{T}$ is 1-1.
  
  To prove surjectivity let $v \in \mathcal{B}$ and $\{ v_n \}_n \subset
  C_c^2 (\mathbb{R}^4)$ such that
  \[ 
       \| v_n - v \|_{L^2_{\ell'}} \rightarrow 0 \quad \tmop{as} \quad n
       \rightarrow \infty .
      \]
  Let $h_n \assign (- \Delta + m^2) v_n$. Then it follows, from the first part
  of Proposition~\ref{prop-reg soln of varphi-bar}, that the elliptic PDE
  \[ 
       (- \Delta + m^2) u_n + G (u_n) = h_n
      \]
  admits a unique solution in $C_{\ell}^2 (\mathbb{R}^4)$. Then we get
  \begin{equation}
      u_n +\mathcal{G} \ast G (u_n) =\mathcal{G} \ast h_n = v_n \quad \Rightarrow       \mathcal{T} (u_n) = v_n .
    \label{T-un-vn}
  \end{equation}
  Next, we prove that $\{ u_n \}_n$ forms a Cauchy sequence in $L_{\lambda,
  \ell'}^2 (\mathbb{R}^4)$ for sufficiently small $\lambda > 0$. By
  multiplying
  \begin{equation}
      u_n - u_m +\mathcal{G} \ast G (u_n) -\mathcal{G} \ast G (u_m) = v_n -   v_m \label{unvn}
  \end{equation}
  by $r_{\lambda, \ell'}^2 (G (u_n) - G (u_m))$ and integrate on $\mathbb{R}^4$ we get
  \begin{align*}
     & \langle u_n - u_m, G (u_n) - G (u_m) \rangle_{\lambda, \ell'} +
       \langle \mathcal{G} \ast G (u_n) -\mathcal{G} \ast G (u_m), G (u_n) - G
       (u_m) \rangle_{\lambda, \ell'}\\
       & \qquad =\langle v_n - v_m, G (u_n) - G (u_m) \rangle_{\lambda, \ell'} . 
  \end{align*}
  Since $G$ is increasing, $\langle u_n - u_m, G (u_n) - G (u_m)
  \rangle_{\lambda, \ell'} \geqslant 0$. Thus, the above implies
  \begin{equation}
      \langle \mathcal{G} \ast G (u_n) -\mathcal{G} \ast G (u_m), G (u_n) - G
      (u_m) \rangle_{\lambda, \ell'} \leqslant \langle v_n - v_m, G (u_n) - G
      (u_m) \rangle_{\lambda, \ell'} .
    \label{unvn-est1}
  \end{equation}
  Thus, taking $\mathfrak{v}= G (u_n) - G (u_m)$ in \eqref{est-poincare}
  (modified version for $\lambda$) yield
  \begin{align}
      & \langle \nabla (\mathcal{G} \ast (G (u_n) - G (u_m))), (\mathcal{G}
      \ast (G (u_n) - G (u_m))) \nabla r_{\lambda, \ell'}^2 \rangle\\
      & + m^2 \| \mathcal{G} \ast (G (u_n) - G (u_m)) \|_{L^2_{\lambda,
      \ell'}}^2 \leqslant  \langle \mathcal{G} \ast (G (u_n) - G (u_m)), G (u_n) - G
      (u_m) \rangle_{\lambda, \ell'} .
     \label{unvn-est2}
  \end{align}
  So the last two estimates together with the Cauchy-Schwartz inequality give
  \begin{align}
      & \langle \nabla (\mathcal{G} \ast (G (u_n) - G (u_m))), (\mathcal{G}
      \ast (G (u_n) - G (u_m))) \nabla r_{\lambda, \ell'}^2 \rangle\\
      & \qquad + m^2 \| \mathcal{G} \ast (G (u_n) - G (u_m)) \|_{L^2_{\lambda,
      \ell'}}^2     \leqslant  \| v_n - v \|_{L^2_{\ell} } \| G (u_n) - G
      (u_m) \|_{L^2_{\lambda, \ell'}} .
    \label{unvn-est3}
  \end{align}
  Consequently, the computation as in \eqref{Tmap-2} gives
  \begin{equation}
      ( m^2 {- 2 \lambda \ell'}^2 ) \| \mathcal{G} \ast (G (u_n)
      - G (u_m)) \|_{L^2_{\lambda, \ell'}}^2
      \leqslant  \| v_n - v \|_{L^2_{\ell'}} \| G (u_n) - G
      (u_m) \|_{L^2_{\lambda, \ell'}} .
    \label{unvn-est5}
  \end{equation}
  Substituting $\mathcal{G} \ast (G (u_n) - G (u_m))$ from \eqref{unvn} into
  \eqref{unvn-est5} followed by the reverse triangle inequality yield
  \begin{align}
    \begin{array}{lll}
      ( m^2 {- 2 \lambda \ell'}^2 ) \| u_n - u_m \|_{L^2_{\lambda,
      \ell'}}^2  \leq  ( m^2 {- 2 \lambda \ell'}^2) \| v_n - v \|_{L^2_{\lambda, \ell'} }
       + \| v_n - v \|_{L^2_{\ell'}} \| G (u_n) - G (u_m)
      \|_{L^2_{\lambda, \ell'}} .
    \end{array} \label{unvn-est6}
  \end{align}
  Since $\| v_n - v \|_{L^2_{\ell'} (\mathbb{R}^4)} \rightarrow 0 \tmop{as} n
  \rightarrow \infty$ and $G$ is bounded, for sufficiently small $\lambda > 0$
  we get that $\{ u_n \}_n$ forms a Cauchy sequence in $L_{\lambda, \ell'}^2
  (\mathbb{R}^4)$. Since $L_{\lambda, \ell'}^2 (\mathbb{R}^4)$ is complete,
  there exists $L_{\lambda, \ell'}^2 (\mathbb{R}^4) \ni u = \lim_{n
  \rightarrow \infty} u_n$. Since $G$ is bounded, $G (u) = \lim_{n \rightarrow
  \infty} G (u_n)$ in $L_{\lambda, \ell'}^2 (\mathbb{R}^4)$. Thus by taking
  limit $n \rightarrow \infty$ in \eqref{T-un-vn} we obtain the existence of
  $u \in L_{\lambda, \ell'}^2 (\mathbb{R}^4)$ such that
  \[ 
       u +\mathcal{G} \ast G (u) = v  \quad \Rightarrow \mathcal{T} (u) = v.
      \]
  So if we show that $u \in \mathcal{B}$ then we are done but that is true
  because,
  \[ 
       \| u \|_{C_{\ell}^0} \leqslant  | \alpha |
       R \int_{\mathbb{R}^4} \mathcal{G} (x - y) r_{\lambda, \ell'} (x) \, \mathd
       x + \| v \|_{C_{\ell}^0},
      \]
  which is finite. Hence $u \in \mathcal{B}$ and we finish the proof of
  bijectivity of $\mathcal{T}$.
\end{proof}

Hence we know that $\mathcal{T}^{- 1}$ exists. Let ${\mathcal{T}^{- 1}}  (V) =
v$ for some $v, V \in \mathcal{B}$. Then $V =\mathcal{T} (v)$ and from \eqref{T-map}, we have that
\begin{equation}
    {\mathcal{T}^{- 1}}  (V) = V - \alpha \mathcal{G} \ast K_R \left( \exp
    \left( {\alpha \mathcal{T}^{- 1}}  (V) - C_{\varepsilon} \right) \right) .
  \label{T-inv-map}
\end{equation}
From here it is clear that, for $V \in \mathcal{B}$, \begin{align}
    \left| {\mathcal{T}^{- 1}}  (V) (x) \right| & \leqslant  | V (x) | +
     \alpha \left| \mathcal{G} \ast K_R \left( \exp \left( {\alpha
     \mathcal{T}^{- 1}}  (V) (x) - C_{\varepsilon} \right) \right) \right|\\
     & \leqslant  | V (x) | + \alpha R \int_{\mathbb{R}^4} | \mathcal{G} (x
     - y) | \, \mathd y.   
\end{align}
Consequently, by the Minkowski inequality for integral we get
\begin{align}
    \left\| {\mathcal{T}^{- 1}}  (V) (x) \right\|_{L^2 (\Omega)}^2 & \leq
     \int_{\Omega} | V (x, \omega) |^2 \, \mathbb{P} (\mathd \omega) + (\alpha R)^2
    \left( \int_{\mathbb{R}^4} \left( \int_{\Omega} | \mathcal{G} (x - y) |^2 \, 
    \mathbb{P} (\mathd \omega) \right)^{1 / 2} \, \mathd y \right)^2\\
    & \leq \int_{\Omega} | V (x, \omega) |^2 \, \mathbb{P} (\mathd \omega) +
    (\alpha R)^2 \left( \int_{\mathbb{R}^4} | \mathcal{G} (x - y) | \, \mathd y
    \right)^2 \backassign  C_R .
  \label{est-invT}
\end{align}

In our next result we show that $\mathcal{T}^{- 1}$ is continuous as well.
\begin{lemma}
  \label{lem-T-inv-continuity}The map $\mathcal{T}^{- 1}$ is continuous on
  $\mathcal{B}.$
\end{lemma}

\begin{proof}
  Let $\{ w_n \}_n \subset \mathcal{B}$ be a sequence converging to some $w
  \in \mathcal{B}$. Let us set $\mathcal{T}^{- 1} w_n \backassign \bar{w}_n$
  and $\mathcal{T}^{- 1} w \backassign \bar{w} .$ We will show that $\bar{w}_n
  \rightarrow \bar{w} $ as $n \rightarrow \infty$ in $\mathcal{B}.$ Note that,
  we have
  \begin{equation}
      \bar{w}_n (x) + \alpha \int_{\mathbb{R}^4} \mathcal{G} (x - y) K_R (\exp(\alpha \bar{w}_n (y) - C_{\varepsilon})) \, \mathd y = w_n (x) .\label{eq-wn-bar-wn}
  \end{equation}
  The first claim in the current proof is that the sequence $\{ \bar{w}_n
  \}_n$ is relatively compact in $\mathcal{B}$. In order to prove this, first
  we show that $\{ \bar{w}_n \}_n$ is uniformly bounded. Since $\{ w_n \}_n$
  is convergent in $\mathcal{B}$ and $\alpha$ is a constant, due to
  \eqref{eq-wn-bar-wn} it is sufficient to show the uniform boundedness
  property for $\left\{ \int_{\mathbb{R}^4} \mathcal{G} (\cdot - y) K_R (\exp
  (\alpha \bar{w}_n (y))) \, \mathd y \right\}_n \subset \mathcal{B}.$ For this
  observe that, by  (47), (48)  of~{\cite{AVG2021}} we have
  \begin{align}
    \int_{\mathbb{R}^4} \mathcal{G} (x - y) K_R (\exp (\alpha \bar{w}_n
       (y))) \, \mathd y & ~\leq ~R {\int_{| z | < 1}}  \left\{ \frac{-
       2}{(4 \pi)^2} \log_+ (| z |) + C_1 \right\} \,  \mathd z\\
       &   \qquad + R \int_{| z | \geqslant 1} C_2 \exp (- C_3 | z |) \, \mathd z,   
  \end{align}
  where the rhs is bounded uniformly in $x$ and $n$. To move further, let us
  set 
  $$ \int_{\mathbb{R}^4} \mathcal{G} (x - y) K_R (\exp (\alpha \bar{w}_n (y) - C_{\varepsilon})) d y \backassign \bar{g}_n   \textrm{ in }   \mathcal{B}. $$ But by its
  structure we know that $\bar{g}_n$ solves the following equation uniquely
  \[ 
       (- \Delta + m^2) \bar{g}_n = K_R (\exp (\alpha \bar{w}_n -
       C_{\varepsilon})) .
      \]
  Thus, 
  \[ 
       \| \bar{g}_n \|_{C_{\ell}^2} ~\lesssim~ \| K_R
       (\exp (\alpha \bar{w}_n - C_{\varepsilon})) \|_{C_{\ell}^0
       } ~\leq~ R .
     \]
  This further implies, due to embedding, see (3.10) in
  {\cite{mourratGlobalWellposednessDynamic2017}}, $B^2_{\infty, \infty, \ell}
  (\mathbb{R}^4) \hookrightarrow B^{1 / 2}_{\infty, \infty, \ell}
  (\mathbb{R}^4)$ and the equivalency of $B^{1 / 2}_{\infty, \infty, \ell}
  (\mathbb{R}^4)$ with $\frac{1}{2}$-H{\"o}lder weighted continuous functions,
  the equicontinuity of $\{ \bar{g}_n \}_n$. Thus, since the uniform topology,
  which space $\mathcal{B}$ has, implies the topology of compact convergence,
  \ the Ascoli--Arzel{\`a} theorem (e.g. see Theorem~47.1 on page~290
  in~{\cite{JM00}}) \ implies the relative compactness of $\{ \bar{w}_n \}_n
  \subset \mathcal{B}$. Let us denote a converging subsequence $\{
  \bar{w}_{n_k} \}_k$ of $\{ \bar{w}_n \}_n$ and set the limit as $\mathcal{B}
  \ni \hat{w} \assign \lim_{k \rightarrow \infty} \bar{w}_{n_k}$. Since $G
  (\cdot) = \alpha K_R (\exp (\alpha \cdot - C_{\varepsilon}))$ is smooth and
  bounded, we have
  \[ 
       \alpha \int_{\mathbb{R}^4} \mathcal{G} (x - y) K_R (\exp (\alpha
       \bar{w}_{n_k} (y) - C_{\varepsilon}))\,  \mathd y ~\rightarrow~ \alpha
       \int_{\mathbb{R}^4} \mathcal{G} (x - y) K_R (\exp (\alpha \hat{w} (y) -
       C_{\varepsilon})) \, \mathd y, \]
  as $k \rightarrow \infty$ and thus passing the limit $k \rightarrow \infty$ in \eqref{eq-wn-bar-wn}
  yield
  \[ 
       \hat{w} (x) + \alpha \int_{\mathbb{R}^4} \mathcal{G} (x - y) K_R (\exp
       (\alpha \hat{w} (y) - C_{\varepsilon})) \, \mathd y = w (x) \quad \Rightarrow~
       \mathcal{T} (\hat{w}) = w.\]
  But since $\mathcal{T}^{- 1} w = \bar{w}$ and $\mathcal{T}$ is bijective, we
  have $\hat{w} = \bar{w} .$ Consequently, any converging subsequence $\{
  \bar{w}_{n_k} \}_k$ converges to $\bar{w}$, which implies the continuity as
  desired. Hence the proof of continuity of $\mathcal{T}^{- 1}$ on
  $\mathcal{B}$ is complete.
\end{proof}

Recall that we aim to prove that, for fixed $\varepsilon$ and $R$,
$\varphi_{\varepsilon, R}$, which solves~\eqref{SPDE-e-R-approx} and has
representation~\eqref{defn-F}, is Malliavin differentiable. Due
to~\eqref{eq:F-and-T}, in order to prove the differentiability of
$\varphi_{\varepsilon, R}$ or say $F$ as the next step we show that the map
$\mathcal{T}^{- 1}$, whose existence and continuity is proved, respectively,
in Lemmata~\ref{lem-T-bijectivity} and~\ref{lem-T-inv-continuity}, is differentiable.
\begin{lemma}
  \label{lemma-T-differentiability}The map $\mathcal{T}^{- 1}$ is
  differentiable and there exists a constant $M > 0$ (depends on $m$) such
  that
  \begin{equation}
      \| (\mathcal{T}_v')^{- 1} \|_{\mathfrak{L} (\mathcal{B}, \mathcal{B})}
      \leqslant M,
        \label{bound-derv-T-inverse}
  \end{equation}
  where $\mathfrak{L} (\mathcal{B}, \mathcal{B})$ is the set of all bounded linear operators from $\mathcal{B}$ to $\mathcal{B}$, uniformly for $v \in \mathcal{B}$.
\end{lemma}
\begin{proof}
  It is straightforward to see that the Gateaux derivative of $\mathcal{T}$ at
  $v \in \mathcal{B}$ in the direction of an arbitrary $w \in \mathcal{B}$, is
  \[ 
       \lim_{t \rightarrow 0} \frac{\mathcal{T} (v + t w) -\mathcal{T} (v)}{t}
       = w + \int_{\mathbb{R}^4} \mathcal{G} (\cdot - y) G' (v (y)) w (y)\, 
       \mathd y, \]
  where recall that $G (v (\cdot)) = \alpha K_R (\exp (\alpha v (\cdot) -
  C_{\varepsilon})) .$ Thus, $\mathcal{T}$ is differentiable. Let us denote by
  $\mathcal{T}_v' (w)$ the derivative of $\mathcal{T}$ at $v \in \mathcal{B}$
  in the direction of $w \in \mathcal{B}$ which is defined above, i.e., \
  \begin{equation}
      \mathcal{T}_v' (w) \assign w + \int_{\mathbb{R}^4} \mathcal{G} (\cdot -
      y) G' (v (y)) w (y) \, \mathd y.
    \label{map-T-prime}
  \end{equation}
  Note that since $G'$ is bounded and $w \in \mathcal{B}$, $\mathcal{T}_v'
  (w)$ is a well-defined element of $\mathcal{B}$. Next, let us fix $v \in
  \mathcal{B}$ in the remaining part of the proof. 
  
  We claim that
  $\mathcal{T}_v'$ is one-one. Indeed, let $\mathcal{T}_v' (w) = 0$ for each
  $w \in \mathcal{B}$ as element of $\mathcal{B}$ then by \eqref{map-T-prime}
  we deduce that $w$ solves the following equation
  \begin{equation}
    \begin{array}{l}
      (- \Delta + m^2) w + G' (v) w = 0.
    \end{array} \label{temp-diffeqn}
  \end{equation}
  It is clear that $w = 0$ is a solution to~\eqref{temp-diffeqn}. From the computation in the proof of Proposition~\ref{prop-reg soln of varphi-bar} and
  Lemma~\ref{lemma_exponentialuniqueness},  we know that~\eqref{temp-diffeqn}
  has a unique solution in $C_{\ell}^2 (\mathbb{R}^4)$, in
  particular $w = 0$ is the unique solution to \eqref{temp-diffeqn}. Thus,
  $\mathcal{T}_v'$ is non-degenerate and $(\mathcal{T}_v')^{- 1} \in
  \mathfrak{L} (\mathcal{B}, \mathcal{B})$ is well-defined. 
  
  We aim
  to show that $(\mathcal{T}_v')^{- 1} \in \mathfrak{L} (\mathcal{B},
  \mathcal{B})$ is uniformly bounded in $v \in \mathcal{B}$. For this let us
  take any $U \in \mathcal{B}$ and $W \assign (\mathcal{T}_v')^{- 1} (U)$. \
  Note that $(\mathcal{T}_v')^{- 1} (U)$ satisfies
  \begin{equation}
      (\mathcal{T}_v')^{- 1} (U) = U - \int_{\mathbb{R}^4} \mathcal{G} (\cdot
      - y) G' (v (y)) (y) (\mathcal{T}_v')^{- 1} (U) \, \mathd y.
    \label{Tinv-diffeq}
  \end{equation}
  Let us consider $C_{c, \ell}^2 (\mathbb{R}^4)$, space of functions in
  $C_{\ell}^2 (\mathbb{R}^4)$ having compact support, as subset of
  $\mathcal{B}$. Let $U \in C_{c, \ell}^2 (\mathbb{R}^4)$ and let $V = (-  \Delta + m^2) U \in \mathcal{B}$. Then from \eqref{Tinv-diffeq} we get that
  $(\mathcal{T}_v')^{- 1} (U)$ satisfy the following equation
  \[ 
       (- \Delta + m^2 + G' (v (y))) (\mathcal{T}_v')^{- 1} (U) = V.
     \]
  Here $G' (v) (\mathcal{T}_v')^{- 1} (U)$ is simply the product of two
  functions $G' (v)$ and $(\mathcal{T}_v')^{- 1} (U)$. Thus, since $G'
  \geqslant 0$, Theorem~5.1 on page~145 in~{\cite{AF75}} implies, for $w \in
  \mathcal{B}$,
  \[ 
       (\mathcal{T}_v')^{- 1} (U) (x) = \tilde{\mathbb{E}}_x \left[
       \int_0^{\infty} e^{- m^2 t} V (B_t) \exp \left( - \int_0^t G' (v (B_s)) \,  \mathd s \right)\,  \mathd t \right],
      \]
  by the Feynman--Kac formula, where $B$ is an $\mathbb{R}^4$-valued Brownian
  motion which starts at $x$ defined on a complete probability space
  $(\tilde{\Omega}, \tilde{\mathfrak{F}}, \tilde{\mathbb{P}})$ and
  $\tilde{\mathbb{E}}_x$ denotes the expectation w.r.t. $\tilde{\mathbb{P}}$.
  But the r.h.s. in above can be estimated as follows to get, since $\exp \left(
  - \int_0^t G' (v (B_s)) \, \mathd s \right)$ is bounded,
  \begin{align*}
    (\mathcal{T}_v')^{- 1} (U) (x)  & ~\leq~ \| V \|_{\mathcal{B}}~
       \tilde{\mathbb{E}}_x \left[ \int_0^{\infty} e^{- m^2 t} (1 + \tilde{\mathbb{E}}_x [| B_t |^2])^{\ell / 2} \, \mathd t \right]\\
       & ~\leq~ \| U \|_{C_{\ell}^2 } \int_0^{\infty} e^{-
       m^2 t} (1 + | x |^2 + t)^{\ell / 2} \, \mathd t.  
  \end{align*}
  This gives
  \begin{align*}
    {\| (\mathcal{T}_v')^{- 1} (U) \|_{\mathcal{B}}}  & ~\leq~ \| U
       \|_{C_{\ell}^2} \sup_{x \in \mathbb{R}^4}
       \int_0^{\infty} e^{- m^2 t} (1 + | x |^2)^{- \ell / 2} (1 + | x |^2 +
       t)^{\ell / 2} \, \mathd t\\
       & ~\leq~ \| U \|_{C_{\ell}^2} \sup_{x \in
       \mathbb{R}^4} \int_0^{\infty} e^{- m^2 t} \, \mathd t,   
  \end{align*}
  which is finite. Consequently, by extension to $\mathcal{B}$, we get that
  there exists a constant $M > 0$ (depends on $m$) such that
  \[ 
       \| (\mathcal{T}_v')^{- 1} \|_{\mathfrak{L} (\mathcal{B}, \mathcal{B})}
       \leqslant M \textrm{ for every } v \in \mathcal{B}.
     \]
\end{proof}

Now we come to an important result of our paper that justifies the Malliavin
differentiability of $\varphi_{\varepsilon, R}$, for fix $\varepsilon$ and
$R$, which solves \eqref{SPDE-e-R-approx}. In other words, the next result
gives the differentiability of $F$ which is defined in \eqref{defn-F}.
\begin{theorem} \label{thm-MallDiff}
  Let us fix $\varepsilon > 0$, $R > 1$ and $x \in
  \mathbb{R}^4$. The solution $\varphi \assign \varphi_{\varepsilon, R}$ to
  \eqref{SPDE-e-R-approx} is such that $\varphi (y) \in \mathbb{D}^{1, 2}$,
  for every $y \in \mathbb{R}^4$. Moreover, the process $\{ D_z \varphi (x), z
  \in \mathbb{R}^4 \}$ satisfies
  \[ 
       D_z \varphi (x) + \alpha \int_{\mathbb{R}^4} \mathcal{G} (x - y) D_z
       \varphi (y) G' (\varphi (y)) \, \mathd y ~=~ (\mathfrak{a}_{\varepsilon}
       \ast \mathcal{G}) (x - z),
      \]
  which is equivalent to
  \[ 
       (\mathcal{L}D_z \varphi) (x) + \alpha G' (\varphi (x)) D_z \varphi (x)
       ~=~ (\mathfrak{a}_{\varepsilon} \ast \delta_z) (x)
       ~=~\mathfrak{a}_{\varepsilon} (x - z) .
      \]
\end{theorem}

\begin{proof}
  Since $\varepsilon > 0$ and $R > 1$ are fixed, we will avoid there explicit
  dependency. The idea of the proof is to show that the conditions of Theorem
  \ref{thm-MallavinDiffAbstract} with $F (\omega) \assign
  \varphi_{\varepsilon, R} (x, \omega)$, are satisfied which will imply the
  conclusions of the current result. By \eqref{est-invT} we have that that
  $\varphi_{\varepsilon, R} (x) \in L^2 (\Omega)$. Indeed,
  \begin{align*}
    \| \varphi_{\varepsilon, R} (x) \|_{L^2  (\Omega)}^2 & =  \|
       \mathcal{T}^{- 1} (\mathcal{G}_{\varepsilon} \ast \xi) (x) \|_{L^2 (\Omega)}^2 \\
       & \leq \int_{\Omega} | (\mathcal{G}_{\varepsilon} \ast \xi) (x)
       |^2 \, \mathbb{P} (\mathd \omega) + (\alpha R)^2 \left(
       \int_{\mathbb{R}^4} | \mathcal{G} (x - y) | \, \mathd y \right)^2,   
  \end{align*}
  which is finite with the bound depends on $R$. Denote by $\mathfrak{G}h$ the
  following defined function
  \[ 
       \mathfrak{G}h (x) = (\mathcal{G}_{\varepsilon} \ast h) (x) \assign
       \int_{\mathbb{R}^4} (\mathfrak{a}_{\varepsilon} \ast \mathcal{G}) (x -
       y) h (y)\,  \mathd y.
     \]
  Note that $\varphi_{0, \varepsilon} (\omega) =\mathfrak{G} \xi$. Observe that, for $h \in L^2_{\ell} (\mathbb{R}^4)$, due to \eqref{eq:F-and-T}
  \begin{equation}
      F [\omega + t i (h)] = \{ \mathcal{T}^{- 1} (\mathfrak{G} \xi +
      t\mathfrak{G}h) \} (x) .
    \label{rel-F-T-inverse}
  \end{equation}
  But, since the above expression is linear in $t$, $F [\omega + t i (h)]$ as   function of $t$ is an absolutely continuous function of $t$. From   Lemma~\ref{lemma-T-differentiability}, we know that $\mathbb{P}$-a.s. $\mathcal{T}_{\varphi_{0, \varepsilon}}'$ exists and non-degenerate and   satisfy $\| (\mathcal{T}_{\varphi_{0, \varepsilon}}')^{- 1} \|_{\mathfrak{L}
  (\mathcal{B}, \mathcal{B})} \leqslant M$, $\mathbb{P}$-a.s. Thus,
  $\mathbb{P}$-a.s. we have
  \[ 
       \lim_{t \rightarrow 0} \frac{F [\omega + t i (h)] - F (\omega)}{t}  ~=~
        \{ (\mathcal{T}^{- 1} )_{\varphi_{0, \varepsilon}}' (\mathfrak{G}h)
       \} (x) .
      \]
  Finally, due to the nice decay property of $\mathfrak{a}_{\varepsilon} \ast
  \mathcal{G}$ and H{\"o}lder inequality, $\mathbb{P}$-a.s. we have
  \[ 
       \| (\mathcal{T}^{- 1} )_{\varphi_{0, \varepsilon}}' (\mathfrak{G}h)
       \|_{L_{\ell}^{\infty}}  ~\leq~  \| (\mathcal{T}^{- 1}
       )_{\varphi_{0, \varepsilon}}' \|_{\mathfrak{L} (\mathcal{B},       \mathcal{B})}  ~ \| h \|_{L^2_{\ell}} \left(
       \int_{\mathbb{R}^4} | (\mathfrak{a}_{\varepsilon} \ast \mathcal{G}) (y)
       r_{\ell}(y) |^2 \, \mathd y \right)^{1 / 2},
     \]
  but this is finite. Now we define an $H$-valued random variable by
  \[ \varsigma : \Omega \ni \omega \mapsto \{ (\mathcal{T}^{- 1}
     )_{\varphi_{0, \varepsilon}}' \mathfrak{G} \} (x) \in H \textrm{ such that } \langle \varsigma (\omega), h \rangle = \{ (\mathcal{T}^{- 1}
     )_{\varphi_{0, \varepsilon}}' (\mathfrak{G}h) \} (x) . \]
  This is well-defined by the Riesz representation theorem and satisfies point~$(2)$ of Theorem~\ref{thm-MallavinDiffAbstract}. Hence, we complete
  the proof of Theorem~\ref{thm-MallDiff}.
\end{proof}

Recall that $\varphi_{\varepsilon, R}$ is the unique solution
to~\eqref{SPDE-e-R-approx}. Proceeding further, we set $\theta_{\varepsilon, R}^z \assign D_z (\varphi_{\varepsilon, R})$. By applying Theorem \ref{thm-MallDiff} to $\varphi_{\varepsilon, R}$, we ascertain that $\varphi_{\varepsilon, R} (x) \in \mathbb{D}^{1, 2}$ and at point $z \in\mathbb{R}^4,$
\begin{equation}
  (\mathcal{L} \theta_{\varepsilon, R}^z) (x) + \alpha^2 K_R' (\exp (\alpha
  \varphi_{\varepsilon, R} (x) - C_{\varepsilon})) \exp (\alpha
  \varphi_{\varepsilon, R} (x) - C_{\varepsilon}) \theta_{\varepsilon, R}^z
  (x) =\mathfrak{a}_{\varepsilon} (x - z) . \label{eqn-theta}
\end{equation}
Here we have also used the chain rule (Proposition~1.2.3
of~{\cite{Nualart2006}}). Since~\eqref{eqn-theta} is linear in
$\theta_{\varepsilon, R}^z$, the Feynman-Kac formula yields
\begin{align}
	\theta_{\varepsilon, R}^z  (x) & = \int_0^{\infty} e^{- m^2 t} ~ \mathbb{E}_x
	\left[\mathfrak{a}_{\varepsilon} (x + B_t - z)  e^{- \int_0^t \alpha^2 \mathfrak{K} (s)  \, \mathd s}
	\right] \, \mathd t, \label{FK}
\end{align}
where
\[ 
     \mathfrak{K} (s) \assign K_R' (\exp (\alpha \varphi_{\varepsilon} (x +
     B_s) - C_{\varepsilon})) \exp (\alpha \varphi_{\varepsilon} (x + B_s) -
     C_{\varepsilon}) .
    \]
Here $\mathbb{E}_x$ is the expectation operator w.r.t. the probability
measure $\mathbb{P}^x$ and  \{$B_t, t \geqslant 0$\}  is a
$\mathbb{R}^4$-valued Brownian motion under $\mathbb{P}^x$ with initial
condition $B_0 = x$. Observe that, for $x_1, x_2 \in \mathbb{R}^4,$
expressions \eqref{lemB.1} followed by \eqref{LandDrelation} yield
\begin{align}
		& \tmop{Cov} (F_1 (f \cdot \varphi_{\varepsilon, R} (\cdot + x_1)), F_2
		(f \cdot \varphi_{\varepsilon, R} (\cdot + x_2))) \nonumber \\
		& =  \mathbb{E} \left[\{ \delta (I - L)^{- 1} D (F_1 (f \cdot
		\varphi_{\varepsilon, R} (\cdot + x_1))) \} \{ \delta (I - L)^{- 1} D (F_2
		(f \cdot \varphi_{\varepsilon, R} (\cdot + x_2))) \} \right] \nonumber\\
		& = \int_{\mathbb{R}^4} \mathbb{E} \left[\{ (I - L)^{- 1} D_z  (F_1 (f \cdot
		\varphi_{\varepsilon, R} (\cdot + x_1))) \} \{ (I - L)^{- 1} D_z  (F_2 (f
		\cdot \varphi_{\varepsilon, R} (\cdot + x_2))) \} \right] \, \mathd z\\
		& \quad + \int_{\mathbb{R}^4} \int_{\mathbb{R}^4} \mathbb{E} \big[D_{z'} \{ (I -
		L)^{- 1} D_z  (F_1 (f \cdot \varphi_{\varepsilon, R} (\cdot + x_1))) \} \times \nonumber \\
		& \qquad \hspace{2 cm}
		D_z \{ (I - L)^{- 1} D_{z'}  (F_2 (f \cdot \varphi_{\varepsilon, R} (\cdot
		+ x_2))) \} \big] \, \mathd z \, \mathd z' .
	 \label{covExpression}
\end{align}
Since, see Corollary B.6 in {\cite{FG2019}} for the proof, for every $F \in
L^2 (\Omega)$ we have
\begin{equation}
  \| D (I - L)^{- 1} F \|_{L^2 (\Omega ; L_{\ell}^2 (\mathbb{R}^4))} \lesssim
  \| F \|_{L^2 (\Omega)}, \label{corB.6}
\end{equation}
we estimate the second term in the r.h.s. of \eqref{covExpression}, by the H{\"o}lder inequality and Proposition~1.2.3 of~{\cite{Nualart2006}} along with Lipschitzness of $F_1$ and $F_2$, as 
\begin{align}
		& \int_{\mathbb{R}^4} \int_{\mathbb{R}^4} \mathbb{E} \big[D_{z'} \{ (I -
		L)^{- 1} D_z  (F_1 (f \cdot \varphi_{\varepsilon, R} (\cdot + x_1))) \} \times \nonumber \\
		& \qquad \hspace{2cm}
		D_z \{ (I - L)^{- 1} D_{z'}  (F_2 (f \cdot \varphi_{\varepsilon, R} (\cdot
		+ x_2))) \} \big] \, \mathd z \, \mathd z' \nonumber \\
		& \lesssim  ~ \int_{\mathbb{R}^4} \| D_z (F_1 (f \cdot \varphi_{\varepsilon,
			R} (\cdot + x_1))) \|_{L^2 (\Omega)} \| D_z (F_2 (f \cdot
		\varphi_{\varepsilon, R} (\cdot + x_2))) \|_{L^2 (\Omega)} \, \mathd z \nonumber \\
		& = ~  \int_{\mathbb{R}^4} \| F_1' (f \cdot \varphi_{\varepsilon, R} (\cdot +
		x_1)) D_z (f \cdot \varphi_{\varepsilon, R} (\cdot + x_1)) \|_{L^2
			(\Omega)} \times  \nonumber \\
		& \qquad  \hspace{2cm}\| F_2' (f \cdot \varphi_{\varepsilon, R} (\cdot + x_2)) D_z (f
		\cdot \varphi_{\varepsilon, R} (\cdot + x_2)) \|_{L^2 (\Omega)} \, \mathd z \nonumber \\
		& \lesssim_{F_1, F_2} ~ \int_{\mathbb{R}^4} \| f \cdot D_z
		(\varphi_{\varepsilon, R} (\cdot + x_1)) \|_{L^2 (\Omega)} \| f \cdot D_z
		(\varphi_{\varepsilon, R} (\cdot + x_2)) \|_{L^2 (\Omega)} \, \mathd z \nonumber \\
		& \leq_f ~ \sup_{\substack{x \in B (x_1, 1) \\ y \in B (x_2, 1)}}
		\int_{\mathbb{R}^4} \left(\mathbb{E} \left[| \theta_{\varepsilon, R}^z (x) |^2 \right]\right)^{1 /
			2} \left(\mathbb{E} \left[| \theta_{\varepsilon, R}^z (y) |^2 \right]\right)^{1 /
			2} \,  \mathd z.
	\label{rhs2Est}
\end{align}

Further, since $(I - L)^{- 1}$ is a bounded operator on $L^2 (\Omega)$, as in \eqref{rhs2Est},  the first integral in r.h.s. of \eqref{covExpression} satisfies
\begin{align}  
    & \int_{\mathbb{R}^4} \mathbb{E} \left[\{ (I - L)^{- 1} D_z  (F_1 (f \cdot
		\varphi_{\varepsilon, R} (\cdot + x_1))) \} \{ (I - L)^{- 1} D_z  (F_2 (f
		\cdot \varphi_{\varepsilon, R} (\cdot + x_2))) \} \right] \, \mathd z \nonumber \\
    & \lesssim_{F_1, F_2, f}  ~ \sup_{\substack{x \in B (x_1, 1) \\ y \in B (x_2, 1)}} 
    \int_{\mathbb{R}^4} \left(\mathbb{E} \left[| \theta_{\varepsilon, R}^z (x) |^2 \right]\right)^{1 /
			2} \left(\mathbb{E} \left[| \theta_{\varepsilon, R}^z (y) |^2 \right]\right)^{1 /
			2} \,  \mathd z.
  \label{rhs1Est}
\end{align}
Thus, substituting \eqref{rhs2Est}-\eqref{rhs1Est} into \eqref{covExpression}
yield
\begin{align}
    & | \tmop{Cov} (F_1 (f \cdot \varphi_{\varepsilon, R} (\cdot + x_1)), F_2
    (f \cdot \varphi_{\varepsilon, R} (\cdot + x_2))) |\\
     & \lesssim_{F_1, F_2, f} \sup_{\substack{x \in B (x_1, 1) \\ y \in B (x_2, 1)}}
    \int_{\mathbb{R}^4} \left(\mathbb{E} \left[| \theta_{\varepsilon, R}^z (x) |^2 \right]\right)^{1 /
			2} \left(\mathbb{E} \left[| \theta_{\varepsilon, R}^z (y) |^2 \right]\right)^{1 /
			2} \,  \mathd z.
  \label{covEst}
\end{align}
Applying the Minkowski inequality for integrals and the Fubini theorem  to \eqref{FK} we further have
\begin{align}
    & \left(\mathbb{E} \left[| \theta_{\varepsilon}^z (x) |^2 \right] \right)^{1 / 2} =  \left( \mathbb{E} \left[ \left| \int_0^{\infty} e^{- m^2 t}
    \mathbb{E}_x \left[ \mathfrak{a}_{\varepsilon} (x + B_t - z) e^{- \alpha^2
    \int_0^t \mathfrak{K} (s) \, \mathd s} \right] \mathd t \right|^2 \right]
    \right)^{1 / 2} \nonumber\\
    & \leq  \mathbb{E}_x \left[ \int_0^{\infty} e^{- m^2 t}
    \mathfrak{a}_{\varepsilon} (x + B_t - z) \left( \mathbb{E} \left[ e^{- 2
    \alpha^2 \int_0^t \mathfrak{K} (s) \, \mathd  s} \right] \right)^{1 / 2} \mathd t    \right] \nonumber \\
    & \lesssim \mathbb{E}_x \left[ \int_0^{\infty} e^{- m^2 t}
    \mathfrak{a}_{\varepsilon} (x + B_t - z) \, \mathd t \right]. 
  \label{E-theta-est}
\end{align}
Note that the r.h.s. of \eqref{E-theta-est} is independent of $R$. Hence, substituting the above into \eqref{covEst} and using the Feynman-Kac formula \eqref{FK} (with zero nonlinearity), we obtain 
\begin{align}
    & | \tmop{Cov} (F_1 (f \cdot \varphi_{\varepsilon, R} (\cdot + x_1)), F_2
    (f \cdot \varphi_{\varepsilon, R} (\cdot + x_2))) | \nonumber \\
    & \lesssim  \sup_{\substack{x \in B (x_1, 1) \\ y \in B (x_2, 1)}}  \int_{\mathbb{R}^4}
    \mathbb{E}_x \left[ \int_0^{\infty} e^{- m^2 t} \mathfrak{a}_{\varepsilon}
    (x + B_t - z) \, \mathd t \right] \mathbb{E}_y \left[ \int_0^{\infty} e^{-
    m^2 t} \mathfrak{a}_{\varepsilon} (y + B_t - z) \, \mathd t \right] \, \mathd  z \nonumber \\
    & = \sup_{\substack{x \in B (x_1, 1)\\ y \in B (x_2, 1)}}  \int_{\mathbb{R}^4} \left(
    \int_{\mathbb{R}^4} \mathcal{G}  (x - u) \mathfrak{a}_{\varepsilon} (u - z)\, \mathd u
     \right) \left( \int_{\mathbb{R}^4} \mathcal{G}  (y - u)
    \mathfrak{a}_{\varepsilon} (u - z) \, \mathd u \right) \, \mathd z \nonumber \\
    & \backassign \sup_{\substack{x \in B (x_1, 1)\\ y \in B (x_2, 1)}} I (x, y) . \label{covEst-1}
\end{align}
To estimate $I (x, y)$, we utilize the following representation of the kernel
$\mathcal{G}$, which is based on the Fourier transform, see pg. 273 of {\cite{AY2002}}, 
\[ (\mathcal{G} \ast \mathfrak{a}_{\varepsilon} (\cdot - z)) (x) = [\mathcal{F}^{- 1}
   (\mathcal{F} ((m^2 - \Delta)^{- 1} (\mathfrak{a}_{\varepsilon} (\cdot -
   z))))] (x), \qquad \forall x, z \in \mathbb{R}^4 . \]
Thus, we deduce that, for all $x, y \in \mathbb{R}^4$,
\begin{align}
    I (x, y) & =  \int_{\mathbb{R}^4} (\mathcal{G} \ast \mathfrak{a}_{\varepsilon}) (x
    - z) (\mathcal{G} \ast \mathfrak{a}_{\varepsilon}) (y - z) \, \mathd z  =  ((m^2 - \Delta)^{- 2} (\mathfrak{a}_{\varepsilon} \ast
    \mathfrak{a}_{\varepsilon})) (x - y) \nonumber \\
    & =  [\mathcal{F}^{- 1} (\mathcal{F} ((m^2 - \Delta)^{- 2}
    (\mathfrak{a}_{\varepsilon} \ast \mathfrak{a}_{\varepsilon})))] (x - y)  =  \int_{\mathbb{R}^4} e^{i z \cdot (x - y)} \frac{((\mathcal{F}
    (\mathfrak{a}_{\varepsilon})) (z))^2}{(m^2 + | z |^2)^2} \,  \mathd z.  \label{Kconv}
\end{align}
Next, we apply a change of variable $z \mapsto A z$, in \eqref{Kconv}, where $A$ represents the rotation matrix on $\mathbb{R}^4$ such that the vector $x - y \in
\mathbb{R}^4_H$ transform to align with one axis, let's say the first axis. Then, with $z =A w$, we have 
\begin{align}
    I (x, y) &  =  \int_{\mathbb{R}^4} \frac{e^{i w \cdot \frac{(\bar{x}_1 -
    \bar{y}_1, 0, 0, 0)}{| x - y |}} ((\mathcal{F}
    (\mathfrak{a}_{\varepsilon})) (A w))^2}{(m^2 | x - y |^2 + | A w |^2)^2} \, 
    \mathd w \nonumber \\
    & =  \int_{\mathbb{R}^3} \int_{\mathbb{R}} \frac{e^{\pm i w_1}
    ((\mathcal{F} (\mathfrak{a}_{\varepsilon})) (A w))^2}{(m^2 | x - y |^2 +     w_1^2 + | w_1 |^2)^2} \, \mathd w_1 \, \mathd w_{1, \perp},
  \label{I3a}
\end{align}
where
\[ 
     A x = (\bar{x}_1, 0, 0, 0), A y = (\bar{y}_1, 0, 0, 0) \infixand w =
     (w_1, w_{1, \perp}) \in \mathbb{R} \times \mathbb{R}^3 .
   \]
Let us only consider the positive sign in $e^{\pm i w_1}$. A similar approach will handle the negative sign case, with the contour $C$ containing $-i$ instead $i$.

First, we compute the integral $\int_{\mathbb{R}} \frac{e^{i w_1} ((\mathcal{F}
(\mathfrak{a}_{\varepsilon})) (A (w_1, w_{1, \perp})))^2}{(m^2 | x - y |^2 +
w_{^1}^2 + | w_{1, \perp} |^2)^2} \mathd w_1$ for fixed $w_{1, \perp} \in
\mathbb{R}^3$ using the residue theorem. We define the contour $C$ that traverses along the real lime from $- a$ to $a$ and then counterclockwise along a semicircle centered at 0 from $- a$ to $a$. Choosing $a \geq 1$ ensures that the point $i (m^2 | x - y |^2 + | w_{1, \perp} |^2)^{1 / 2}$ lies within the contour. Now, consider the contour integral
\begin{equation}
  \int_C \frac{e^{i w_1} ((\mathcal{F} (\mathfrak{a}_{\varepsilon})) (A (w_1,
  w_{1, \perp})))^2}{(m^2 | x - y |^2 + w_{^1}^2 + | w_{1, \perp} |^2)^2}
  \mathd w_1 . \label{contourInt}
\end{equation}
Since the integrand in \eqref{contourInt} has singularities at $\pm i (m^2 |
x - y |^2 + | w_{1, \perp} |^2)^{1 / 2}$ with multiplicity 2, by the residue theorem, for fixed $w_{1, \perp} \in \mathbb{R}^3$, we have
\begin{align}
  \int_C \frac{e^{i w_1} ((\mathcal{F} (\mathfrak{a}_{\varepsilon})) (A (w_1,
  w_{1, \perp})))^2}{(m^2 | x - y |^2 + w_{^1}^2 + | w_{1, \perp} |^2)^2}
  \, \mathd w_1  = - 2 \pi i (\tmop{Res} (f (w_1), w_1 = i (m^2 | x - y |^2 + |
  w_{1, \perp} |^2)^{1 / 2})), \label{R1}
\end{align}
where 
$$ f (w_1) = \frac{e^{i w_1} ((\mathcal{F} (\mathfrak{a}_{\varepsilon}))
(A (w_1, w_{1, \perp})))^2}{(m^2 | x - y |^2 + w_{^1}^2 + | w_{1, \perp}
|^2)^2}$$ 
and
\begin{align}
    & \tmop{Res} (f (w_1), w_1 = i (m^2 | x - y |^2 + | w_{1, \perp} |^2)^{1
    / 2}) \nonumber \\
    & = ~ \lim_{z \rightarrow i (m^2 | x - y |^2 + | w_{1, \perp} |^2)^{1 / 2}}
    \frac{\mathd}{\mathd w_1} \left[ \frac{e^{i w_1} ((\mathcal{F}
    (\mathfrak{a}_{\varepsilon})) (A (w_1, w_{1, \perp})))^2}{(w_1 + i (m^2 |
    x - y |^2 + | w_{1, \perp} |^2)^{1 / 2})^2} \right].
  \label{eqn-residue}
\end{align}
Here, with abbreviated notation
$\mathcal{F} (\mathfrak{a}_{\varepsilon}) = (\mathcal{F}
(\mathfrak{a}_{\varepsilon})) (A (w_1, w_{1, \perp}))$, we find that
\begin{align}
    & \frac{\mathd}{\mathd w_1} \left[ \frac{e^{i w_1} ((\mathcal{F}
     (\mathfrak{a}_{\varepsilon})) (A (w_1, w_{1, \perp})))^2}{(w_1 + i (m^2 |
     x - y |^2 + | w_{1, \perp} |^2)^{1 / 2})^2} \right]\\
     & = ~ \frac{e^{i w_1} \mathcal{F} (\mathfrak{a}_{\varepsilon}) \{
     [i\mathcal{F} (\mathfrak{a}_{\varepsilon}) + 2 (\mathcal{F} (- i x_1
     \mathfrak{a}_{\varepsilon} (x)))] (w_1 + i (m^2 | x - y |^2 + | w_{1,
     \perp} |^2)^{1 / 2}) - 2\mathcal{F} (\mathfrak{a}_{\varepsilon}) \}}{(w_1
     + i (m^2 | x - y |^2 + | w_{1, \perp} |^2)^{1 / 2})^3}   . 
\end{align}
Hence, substituting the above expression in \eqref{eqn-residue} and then taking the limit as $a
\rightarrow \infty $ in \eqref{R1}, we obtain, for fixed $w_{1, \perp} \in\mathbb{R}^3$,
\begin{align}
    & \int_{\mathbb{R}} \frac{e^{i w_1} ((\mathcal{F}
    (\mathfrak{a}_{\varepsilon})) (A (w_1, w_{1, \perp})))^2}{(m^2 | x - y |^2
    + w_{^1}^2 + | w_{1, \perp} |^2)^2} \, \mathd w_1 \nonumber\\
    & = ~ \frac{\pi}{4} \frac{e^{- (\mathfrak{w}_{1, \perp} (x, y))^{1 / 2}}
    \mathcal{F} (\mathfrak{a}_{\varepsilon}) \{ [i\mathcal{F}
    (\mathfrak{a}_{\varepsilon}) + 2 (\mathcal{F} (-
    i\mathfrak{a}_{\varepsilon} (x)))] (2 i (\mathfrak{w}_{1, \perp} (x,
    y))^{1 / 2}) - 2\mathcal{F} (\mathfrak{a}_{\varepsilon}) \}}{{(m^2 | x - y
    |^2 + | w_{1, \perp} |^2)^{3 / 2}} } .
  \label{I4}
\end{align}
where we set  $m^2 | x - y |^2 + | w_{1, \perp} |^2 \backassign \mathfrak{w}_{1, \perp} (x, y) .$ 
Consequently,
\begin{align}
    I (x, y) & \lesssim \int_{\mathbb{R}^3} \frac{e^{- (\mathfrak{w}_{1, \perp} (x,
    y))^{1 / 2}} | (\mathcal{F} (\mathfrak{a}_{\varepsilon})) (A (w_1, w_{1,
    \perp})) |^2}{\mathfrak{w}_{1, \perp} (x, y)} \mathd w_{1, \perp} \nonumber \\
    & \quad + \int_{\mathbb{R}^3} \frac{e^{- (\mathfrak{w}_{1, \perp} (x, y))^{1 /
    2}} | (\mathcal{F} (\mathfrak{a}_{\varepsilon})) (A (w_1, w_{1, \perp})) |
    | (\mathcal{F} (- i x_1 \mathfrak{a}_{\varepsilon} (x) \nobracket) (A
    (w_1, w_{1, \perp})) |}{\mathfrak{w}_{1, \perp} (x, y)} \mathd w_{1,
    \perp} \nonumber \\
    & \quad + \int_{\mathbb{R}^3} \frac{e^{- (\mathfrak{w}_{1, \perp} (x, y))^{1 /
    2}} | (\mathcal{F} (\mathfrak{a}_{\varepsilon})) (A (w_1, w_{1, \perp}))
    |^2}{{(\mathfrak{w}_{1, \perp} (x, y))^{3 / 2}} } \mathd w_{1, \perp} \nonumber \\
    &  \backassign I_1 (x, y) + I_2 (x, y) + I_3 (x, y) .
  \label{I5}
\end{align}

Since $- i x_1 \mathfrak{a}_{\varepsilon} (x)$, where $x= (x_1,x_2,x_3,x_4)$,  is also a smooth and compactly supported function, it suffices to estimate $I_1$ and $I_3$ in \eqref{I5}. For $I_1$, using estimate \eqref{est-PWS}, for $N =1$, where we let $w_{1, \perp} = m u | x - y |$, we have 
\begin{align}
    I_1 (x, y) & \lesssim  C_N \int_{\mathbb{R}^3} \frac{e^{-
    (\mathfrak{w}_{1, \perp} (x, y))^{1 / 2}} (1 + | A (i (\mathfrak{w}_{1,
    \perp} (x, y))^{1 / 2}, w_{1, \perp}) |)^{- 2 N}}{\mathfrak{w}_{1, \perp}
    (x, y)} \, \mathd w_{1, \perp} \nonumber \\
    & \lesssim_N  \int_{\mathbb{R}^3} \frac{e^{- (\mathfrak{w}_{1, \perp}
    (x, y))^{1 / 2}} (1 + | i (\mathfrak{w}_{1, \perp} (x, y))^{1 / 2}, w_{1,
    \perp} |^2)^{- N}}{\mathfrak{w}_{1, \perp} (x, y)} \mathd w_{1, \perp} \nonumber \\
    \dela{& =  \int_{\mathbb{R}^3} \frac{e^{- \left( m^2 | x - y |^2 + m^2 | x - y
    |^2 |u|^2  \right)^{1 / 2}} (1 + m^2 | x - y |^2 + 2 | x -
    y |^2 | u |^2)^{- N}}{m^2 | x - y |^2 + m^2 | x - y |^2 {| u |^2} } m^3 |
    x - y |^3 \mathd u  \quad \nonumber\\}
    & \lesssim  e^{- m | x - y |} \int_{\mathbb{R}^3} \frac{m | x - y | e^{-
    | u |}}{\left( 1 + |u|^2  \right) (m^2 | x - y |^2 + 2 m^2
    | x - y |^2 | u |^2)^N} \, \mathd u \nonumber \\
    & \lesssim e^{- m | x - y |} \int_{\mathbb{R}^3} \frac{e^{- |u |}}{\left( 1 + | u|^2 \right)^2} \, \mathd u. 
  \label{I6}
\end{align}
For $I_3$ in \eqref{I5},  we can perform similar computation and obtain, 
\begin{align}
    I_3 (x, y) & \lesssim  C_N \int_{\mathbb{R}^3} \frac{e^{-
    (\mathfrak{w}_{1, \perp} (x, y))^{1 / 2}} (1 + | A (i (\mathfrak{w}_{1,
    \perp} (x, y))^{1 / 2}, w_{1, \perp}) |)^{- 2 N}}{(\mathfrak{w}_{1, \perp}
    (x, y))^{3 / 2}} \,  \mathd w_{1, \perp} \nonumber\\
    & \lesssim  \int_{\mathbb{R}^3} \frac{e^{- \left( m^2 | x - y |^2 + m^2
    | x - y |^2 | u |^2  \right)^{1 / 2}} (1 + m^2 | x - y |^2 +
    2 | x - y |^2 | u |^2)^{- N}}{\left( m^2 | x - y |^2 + | x - y |^2 |
    u|^2  \right)^{3 / 2}} m^3 | x - y |^3 \, \mathd u  \nonumber \\
    & \lesssim  e^{- m | x - y |} \int_{\mathbb{R}^3} \frac{e^{- | u |}}{(m
    | x - y |)^2 \left( 1 + | u |^2  \right)^{3 / 2} (1 + | u
    |^2)} \, \mathd u \nonumber \\
    & \lesssim  e^{- m | x - y |}\int_{\mathbb{R}^3} (1 + | u | )^{- 5} \mathd u . 
  \label{I7}
\end{align}
Thus, substituting \eqref{I6}-\eqref{I7} into  \eqref{I5}, yields
\[ 
     I (x, y) \lesssim e^{- m | x - y |}.
   \]
Further, by substituting this into \eqref{covEst-1}, we  can make the estimation under the condition $l = | x_1 - x_2 | > 2$ as
\begin{equation}
    | \tmop{Cov} (F_1 (f \cdot \varphi_{\varepsilon, R} (\cdot + x_1)), F_2
    (f \cdot \varphi_{\varepsilon, R} (\cdot + x_2))) | \lesssim_{F_1, F_2,
    f}  e^{- m l} .
  \label{covEst-final-approx}
\end{equation}
The complementary case of $| x_1 - x_2 | \leqslant 2$ is  straightforward, akin to the coupling approach.  
Therefore, since \eqref{covEst-final-approx}  holds uniformly in $\varepsilon$ and $R$, we conclude the proof of Theorem~\ref{thm-main} by first taking the limit taking $R \rightarrow \infty$ and then letting $\varepsilon \rightarrow 0$.

\appendix\section{Auxiliary results}\label{sec-appendix}
The first result in this section is about the solution theory to
equation~\eqref{SPDE-Kregapprox0}.
\begin{proposition}
  \label{prop-reg soln of varphi-bar}For given $\varepsilon \in (0, 1)$ and $R
  \geqslant 1$, there exists a $\bar{\varphi}_{\varepsilon, R} \in
  C_{\ell}^2 (\mathbb{R}^4) \assign B^2_{\infty, \infty, \ell}
  (\mathbb{R}^4)$ which solves \eqref{SPDE-Kregapprox0}. Moreover, $\alpha
  \bar{\varphi}_{\varepsilon, R} \leqslant 0.$
\end{proposition}
\begin{proof}
  Let us introduce the following map
  \begin{equation}
      \mathcal{K} (\bar{\varphi}_{\varepsilon, R}, \eta_{\varepsilon}) \assign
      - \alpha (- \Delta + m^2)^{- 1} (K_R (\exp (\alpha
      \bar{\varphi}_{\varepsilon, R}) \eta_{\varepsilon})) .
    \label{defn-K}
  \end{equation}
  We first show that there exists a solution $\bar{\varphi}_{\varepsilon, R}
  \in B^2_{\infty, \infty, \ell} (\mathbb{R}^4)$ to the equation
  \[ 
       \bar{\varphi}_{\varepsilon, R} =\mathcal{K}
       (\bar{\varphi}_{\varepsilon, R}, \eta_{\varepsilon}) .
      \]
  We aim to use Schaefer's fixed-point theorem (see Theorem 4 in Section 9.2
  of Chapter 9 of {\cite{Evans98}}) to prove the claim. In order to do this we
  have to prove that $\mathcal{K}$ is continuous in
  $\bar{\varphi}_{\varepsilon, R}$, that it maps any bounded set into a
  compact set and that the set of solutions to the equations
  \[ 
       \bar{\varphi} ~=~ \lambda \mathcal{K} (\bar{\varphi}, \eta)
      \]
  is bounded uniformly for all $0 \leqslant \lambda \leqslant 1$. The
  continuity of $\mathcal{K}$ is an easy consequence of continuity of $(-
  \Delta + m^2)^{- 1}$ from $B^0_{\infty, \infty, \ell} (\mathbb{R}^4)$ into
  $B^2_{\infty, \infty, \ell} (\mathbb{R}^4)$ and properties of functions $K_R$ and  $\exp$. The map $\mathcal{K}$ is compact because the Schauder
  estimates and embedding $L^{\infty}_{\ell} (\mathbb{R}^4)
  \longhookrightarrow B^0_{\infty, \infty, \ell} (\mathbb{R}^4)$ imply
  \begin{equation}
    \begin{array}{lllll}
      \| \mathcal{K} (\bar{\varphi}, \eta_{\varepsilon}) \|_{B^2_{\infty,
      \infty, \ell}} & \leqslant & | \alpha |~ \| K_R (\exp
      (\alpha \bar{\varphi}) \eta) \|_{L^{\infty}_{\ell}} &
      \leqslant & R | \alpha |,
    \end{array} \label{unif-1}
  \end{equation}
  and the immersion $B^2_{\infty, \infty, \ell} (\mathbb{R}^4)
  \longhookrightarrow B^{2 - \delta}_{\infty, \infty, \ell + \delta'}
  (\mathbb{R}^4)$ is compact, see Proposition~52 of~{\cite{AVG2021}}. Finally
  the uniform boundedness in $\mathlambda$ follows from inequality
  \eqref{unif-1}. Thus, by Schaefer's fixed-point theorem there exists a fixed point of $\bar{\varphi} =\mathcal{K} (\bar{\varphi}, \eta)$ in $B^{2 -
  \delta}_{\infty, \infty, \ell + \delta'} (\mathbb{R}^4)$. Let us call it
  $\bar{\varphi}_{\varepsilon, R}$. Further note that, since
  $\bar{\varphi}_{\varepsilon, R}$ is a fixed point, \eqref{unif-1} also give
  \begin{equation}
    \begin{array}{lllll}
      \| \bar{\varphi}_{\varepsilon, R} \|_{B^2_{\infty, \infty, \ell}} & = & \| \mathcal{K} (\bar{\varphi}_{\varepsilon, R},
       \eta_{\varepsilon}) \|_{B^2_{\infty, \infty, \ell}}       & \leqslant & R | \alpha |.
    \end{array} \label{unif-1temp}
  \end{equation}
  Thus, $\bar{\varphi}_{\varepsilon, R} \in B^2_{\infty, \infty, \ell}
  (\mathbb{R}^4)$. Hence the first part of the proof.
  
  Next, since $\bar{\varphi}_{\varepsilon, R} \in C_{\ell}^2
  (\mathbb{R}^4)$, $\bar{\varphi}_{\varepsilon, R} \in L_{\ell}^{\infty}
  (\mathbb{R}^4)$. Let us define, for $x \in \mathbb{R}^4$,
  \[      r_{\ell, \theta} (x) \assign (1 + \theta | x |^2)^{- \ell}, \quad 
       \bar{\varphi}_{\varepsilon, R, \alpha} \assign \alpha
       \bar{\varphi}_{\varepsilon, R} \quad \infixand \quad \psi \assign r_{\ell, \theta}
       \bar{\varphi}_{\varepsilon, R, \alpha},
      \]
  where $\ell$ is chosen such that the first part of the current proposition
  holds valid. Note that from the first part, $\psi$ is bounded and locally
  belongs to $C^2 (\mathbb{R}^4)$. Assume for the moment that $\psi$
  has a global maximum and attains its maximum value at $\hat{x} .$ Then,
  since $\hat{x}$ is a critical point,
  \begin{equation}
      0 = \nabla \psi = \bar{\varphi}_{\varepsilon, R, \alpha} \nabla r_{\ell,
      \theta} + r_{\ell, \theta} \nabla \bar{\varphi}_{\varepsilon, R,
      \alpha}, \label{ineq:1stDerv}
  \end{equation}
  and, thus, by the second derivative test,
  \begin{equation}
      0 \leqslant - \Delta \psi = - r_{\ell, \theta} \Delta
      \bar{\varphi}_{\varepsilon, R, \alpha} - \bar{\varphi}_{\varepsilon, R,
      \alpha} \Delta r_{\ell, \theta} + 2 \frac{| \nabla r_{\ell, \theta}
      |^2}{r_{\ell, \theta}} \bar{\varphi}_{\varepsilon, R, \alpha}. \label{ineq:2ndDerv}
  \end{equation}
  But
  \[ \Delta \bar{\varphi}_{\varepsilon, R, \alpha} = \alpha \Delta
     \bar{\varphi}_{\varepsilon, R} = \alpha m^2 \bar{\varphi}_{\varepsilon,
     R} + \alpha^2 K_R (\exp (\alpha \bar{\varphi}_{\varepsilon, R})
     \eta_{\varepsilon}), \]
  so, from \eqref{ineq:2ndDerv}
  \[ 
       \alpha m^2 \bar{\varphi}_{\varepsilon, R} + \alpha^2 K_R (\exp (\alpha
       \bar{\varphi}_{\varepsilon, R}) \eta_{\varepsilon}) \leqslant - \alpha
       \left[ \frac{\Delta r_{\ell, \theta}}{r_{\ell, \theta}} - 2 \frac{|
       \nabla r_{\ell, \theta} |^2}{(r_{\ell, \theta})^2} \right]
       \bar{\varphi}_{\varepsilon, R} .
      \]
  Since $\alpha^2 K_R (\exp (\alpha \bar{\varphi}_{\varepsilon, R})
  \eta_{\varepsilon}) \geqslant 0$, we get
  \begin{equation}
      \alpha m^2 \bar{\varphi}_{\varepsilon, R} + \alpha \left[ \frac{\Delta
      r_{\ell, \theta}}{r_{\ell, \theta}} - 2 \frac{| \nabla r_{\ell, \theta}
      |^2}{(r_{\ell, \theta})^2} \right] \bar{\varphi}_{\varepsilon, R}
      \leqslant 0. \label{ineq:varphi-r-m}
  \end{equation}
  But note that due to the choice of weight $r_{\ell, \theta}$, we can choose
  $\theta > 0$ such that
  \[ \left| \frac{- r_{\ell, \theta} \Delta r_{\ell, \theta} + 2 | \nabla
     r_{\ell, \theta} |^2}{r_{\ell, \theta}^2} \right| < m^2 . \]
  Hence, with the choice of $\theta$ from \eqref{ineq:varphi-r-m} we have
  \[ \alpha \bar{\varphi}_{\varepsilon, R} \left[ m^2 - \left\{ - \frac{\Delta
     r_{\ell, \theta}}{r_{\ell, \theta}} + 2 \frac{| \nabla r_{\ell, \theta}
     |^2}{r_{\ell, \theta}^2} \right\} \right] \leqslant 0, \]
  where the quantity in curly bracket is positive. Thus, the above is only possible if $\alpha \bar{\varphi}_{\varepsilon, R} \leqslant 0$.  Hence
  we have prove the result in the case $\bar{\varphi}_{\varepsilon, R}$
  attains its maximum. The case when it does not, can be taken care as  explained in Lemma~2.8 in~{\cite{GH2019}}. This completes the proof. 
\end{proof}

Now we move to the uniqueness of above constructed solution. The approach to
prove the next result is closely related to the proof of Lemma~31
in~{\cite{AVG2021}}.

\begin{lemma}
  \label{lemma_exponentialuniqueness}For given $\varepsilon \in (0, 1)$ and $R
  \geqslant 1$, the solution to equation \eqref{SPDE-Kregapprox0} is unique in
  $C_{\ell}^2 (\mathbb{R}^4)$. 
\end{lemma}

\begin{proof}
  Let $\varepsilon \in (0, 1)$ and $R \geqslant 1$ be fixed parameters. We will omit explicit mention of them for the remainder of the proof. Consider $J : \mathbb{R} \rightarrow \mathbb{R}$,  a smooth, bounded, strictly increasing function
  such that $J (0) = 0$ and $J (- x) = - J (x)$. Further, let
  $\bar{\varphi}_1$ and $\bar{\varphi}_2$ be two solutions to equation   \eqref{SPDE-Kregapprox0}. Since they are smooth, $J (\bar{\varphi}_1 -
  \bar{\varphi}_2) \in C_{\ell}^2 (\mathbb{R}^4)$ implying that
  $r_{\ell'} (\lambda z) J (\bar{\varphi}_1 - \bar{\varphi}_2) \in
  C_{\ell}^2 (\mathbb{R}^4)$ for $\ell' > 0$ sufficiently large enough and any $\lambda > 0$. This implies that, where $\langle \cdot, \cdot \rangle$ denotes
  just the $L^2 (\mathbb{R}^4)$-inner product,
  \[ \langle r_{\ell'} (\lambda z) J (\bar{\varphi}_1 - \bar{\varphi}_2), (-
     \Delta + m^2) (\bar{\varphi}_1 - \bar{\varphi}_2
     -\mathcal{K}(\bar{\varphi}_1, \eta) +\mathcal{K}(\bar{\varphi}_2, \eta))
     \rangle = 0, \]
  where $\mathcal{K}$ is defined in \eqref{defn-K}. 
  We claim that the inequality
  \begin{equation}
      \langle r_{\ell'} (\lambda z) J (\bar{\varphi}_1 - \bar{\varphi}_2), (-
      \Delta + m^2) (\bar{\varphi}_1 - \bar{\varphi}_2) \rangle \geqslant C
      \int r_{\ell'} (\lambda z) J (\bar{\varphi}_1 - \bar{\varphi}_2)
      (\bar{\varphi}_1 - \bar{\varphi}_2) \, \mathd z,
    \label{eq:inequalityuniqueness}
  \end{equation}
  holds for sufficiently small $\lambda > 0$ and some constant $C > 0$. Indeed, we have
    \begin{align*}
        & \langle r_{\ell'} (\lambda z) J (\bar{\varphi}_1 - \bar{\varphi}_2), (-
       \Delta + m^2) (\bar{\varphi}_1 - \bar{\varphi}_2) \rangle  = \int r_{\ell'} (\lambda z) J' (\bar{\varphi}_1 -
       \bar{\varphi}_2) | \nabla \bar{\varphi}_1 - \nabla \bar{\varphi}_2 |^2
       \, \mathd z \\
       & \qquad \hspace{4em} + \lambda \int \nabla r_{\ell'} (\lambda \hat{z}) J
       (\bar{\varphi}_1 - \bar{\varphi}_2) \cdot (\nabla \bar{\varphi}_1 -
       \nabla \bar{\varphi}_2) \, \mathd z\\
       & \qquad \hspace{4em} + m^2  \int r_{\ell'} (\lambda \hat{z}) J
       (\bar{\varphi}_1 - \bar{\varphi}_2) (\bar{\varphi}_1 - \bar{\varphi}_2)\,       \mathd z\\
       & \qquad \qquad \geqslant - \lambda^2  \int (\Delta r_{\ell'} (\lambda
       \hat{z})) J^{- 1} (\bar{\varphi}_1 - \bar{\varphi}_2) \, \mathd z + m^2 
       \int r_{\ell'}  (\lambda \hat{z}) J (\bar{\varphi}_1 - \bar{\varphi}_2)
       (\bar{\varphi}_1 - \bar{\varphi}_2) \, \mathd z\\
       & \qquad \qquad \geqslant \int \left( m^2 - \left| \frac{\lambda^2 \Delta
       r_{\ell'}}{r_{\ell'}} \right| \right) r_{\ell'} (\lambda z) J
       (\bar{\varphi}_1 - \bar{\varphi}_2) (\bar{\varphi}_1 - \bar{\varphi}_2)     \, \mathd z,     
    \end{align*}
  where $J^{- 1} (t) = \int_0^t J (\tau) \mathd \tau$. By selecting a sufficiently small $\lambda > 0$, we get the claim. For the first inequality we utilize  the
  following fact: 
  \begin{align*}
    \int \nabla r_{\ell'} (\lambda z) J (\bar{\varphi}_1 - \bar{\varphi}_2) & 
       \cdot (\nabla \bar{\varphi}_1 - \nabla \bar{\varphi}_2) \, \mathd z = \int
       \nabla r_{\ell'} (\lambda z) \nabla J^{- 1} (\bar{\varphi}_1 -
       \bar{\varphi}_2) \, \mathd z \\
       & = - \lambda \int \Delta r_{\ell'} (\lambda z) J^{- 1} (\bar{\varphi}_1
       - \bar{\varphi}_2) \, \mathd z,   
  \end{align*}
  which holds true since $J^{- 1}$ is a Lipschitz function satisfying $J^{- 1} (0)= 0$. Additionally, we exploit the increasing behavior of $J$ to establish 
  $J^{- 1} (t) \leqslant t J (t) $.
  
  The next claim is that $\langle r_{\ell'} (\lambda z) J
  (\bar{\varphi}_1 - \bar{\varphi}_2), (- \Delta + m^2)
  (-\mathcal{K}(\bar{\varphi}_1, \eta) +\mathcal{K}(\bar{\varphi}_2, \eta))
  \rangle \geqslant 0$. To demonstrate this, we have 
  \begin{align}
      & \langle r_{\ell'} (\lambda z) J (\bar{\varphi}_1 - \bar{\varphi}_2), (-
      \Delta + m^2) (-\mathcal{K}(\bar{\varphi}_1, \eta)
      +\mathcal{K}(\bar{\varphi}_2, \eta)) \rangle \\
      & \quad = \int r_{\ell'} (\lambda z) (\alpha (K_R (\exp (\alpha \bar{\varphi}_1)
      \eta)) - \alpha (K_R (\exp (\alpha \bar{\varphi}_2) \eta))) J
      (\bar{\varphi}_1 - \bar{\varphi}_2) \mathd z ~\geqslant~ 0,
     \label{eq:claim2}
  \end{align}
  where we use the fact that $(\alpha (K_R (\exp (\alpha t_1) \eta)) - \alpha (K_R
  (\exp (\alpha t_2) \eta))) \cdot J (t_1 - t_2)$ is positive since both  $\alpha (K_R (\exp (\alpha \cdot) \eta))$ and $J$ are increasing functions   and $J (0) = 0$. 
  Thus, combining inequalities \eqref{eq:inequalityuniqueness} and
  \eqref{eq:claim2} we deduce that
  \[ \int r_{\ell'} (\lambda z) J (\bar{\varphi}_1 - \bar{\varphi}_2)
     (\bar{\varphi}_1 - \bar{\varphi}_2) \mathd z \leqslant 0, \]
  which implies $\bar{\varphi}_1 - \bar{\varphi}_2 = 0$, since $J$ is a strictly increasing function. Consequently, the proof of uniqueness is established.
\end{proof}

For the next result assume that $\Delta$ is the $d$-dimensional Laplacian. Let
us denote the kernel representation of $\mathcal{L}^{- 1}$ by
\[ (\mathcal{L}^{- 1} \phi) (x) = \int_{\mathbb{R}^d} \mathcal{G}  (x - y) \phi (y) \, 
   \mathd y, \qquad \phi \in \mathcal{S} . \]
\begin{lemma}
  \label{lem-AY2002} $\mathcal{G}$ has the following integral representation
  \[ \mathcal{G} (x) = \frac{1}{{(4 \pi)^{\frac{d}{2}}} } \int_0^{\infty} \exp \left\{ -
     \frac{| x |^2}{4 s} - m^2 s \right\} s^{\frac{- d}{2}} \, \mathd s, \qquad
     x \in \mathbb{R}^d . \]
  Moreover, there exist some constants $C_1, C_2 > 0$ such that the following
  holds:
  \begin{enumeratenumeric}
    \item if $d > 2$ then
    \[ \mathcal{G} (x) \leq C_1 | x |^{- d + 2}  \textrm{ \quad \tmop{if}} | x | < 1
       \textrm{\qquad \infixand}  \qquad C_1 e^{- C_2 | x | } \quad 
       \textrm{\tmop{if}} | x | \geq 1; \]
    \item if $d < 2$ then
    \[ \mathcal{G} (x) \leq C_1 | x |^{- d + 2}  \textrm{ \quad \tmop{for} \qquad } x
       \in \mathbb{R}^d; \]
    \item if $d = 2$ then
    \[ \mathcal{G} (x) \leq C_1 - \frac{2}{(4 \pi)^{\frac{d}{2}} \Gamma \left(
       \frac{d}{2} \right)} \log (| x |)  \textrm{ \quad \tmop{if}} | x | < 1
       \textrm{\qquad \infixand}  \qquad C_1 e^{- C_2 | x | } \quad 
       \textrm{\tmop{if}} | x | \geq 1. \]
  \end{enumeratenumeric}
\end{lemma}

\begin{proof}
  See Proposition~A.1 in~{\cite{AY2002}}.
\end{proof}

The next result is well-known in the literature.

\begin{theorem}[Paley--Wiener--Schwartz]
  \label{thm-PWS} For any $d \in \mathbb{N}$, the vector space $C_c^{\infty}(\mathbb{R}^d)$, comprising compactly supported smooth functions on $\mathbb{R}^d$, is isomorphic, via the Fourier transform, to the space of entire functions $F$ on $\mathbb{C}^d$ satisfying the following condition: there exists a positive
  real number $B$ such that for every integer $N > 0$, there is a real number $C_N > 0$ such that
  \begin{equation}
    | F (\xi) | \leqslant C_N (1 + | \xi |)^{-
    N} e^{B | \tmop{Im} (\xi) |}, \quad \forall \xi \in \mathbb{C}^d. \label{est-PWS}
  \end{equation}
  This implies that for any $u \in C_c^{\infty} (\mathbb{R}^d)$, there exists an entire function $F = \hat{u}$ satisfying the above estimate.
\end{theorem}

Finally we need the following Besov embedding.

\begin{theorem}
  \label{thm-BesovEmbedding}Consider $p_1, p_2, q_1, q_2 \in [1, \infty], s_1
  > s_2$ and $\ell_1, \ell_2 \in \mathbb{R}$ such that
  \[       \ell_1 \leqslant \ell_2 \quad \infixand \quad s_1 - \frac{d}{p_1} \geqslant s_2 -
       \frac{d}{p_2},
      \]
  then $B_{p_1, q_1, \ell_1}^{s_1} (\mathbb{R}^d)$ is continuously embedded in
  $B_{p_2, q_2, \ell_2}^{s_2} (\mathbb{R}^d)$. And if $\ell_1 < \ell_2$ and
  $s_1 - \frac{d}{p_1} > s_2 - \frac{d}{p_2}$ then the embedding $B_{p_1, q_1,
  \ell_1}^{s_1} (\mathbb{R}^d) \hookrightarrow B_{p_2, q_2, \ell_2}^{s_2}
  (\mathbb{R}^d)$ is compact. 
\end{theorem}

\begin{proof}
  See Theorem~6.7 in~{\cite{TriebelIII}}.
\end{proof}

\end{document}